\documentclass[11pt,onecolumn,twoside]{IEEEtran}
\usepackage{hyperref}

\usepackage{amsfonts,amssymb,amsmath,amsthm,dsfont}

\usepackage{subfigure}
\usepackage{fullpage, color, upgreek}
\usepackage{graphicx}
\makeatletter
\let\NAT@parse\undefined
\makeatother
\usepackage[sort&compress, numbers]{natbib}

\usepackage{algorithm,algorithmic}
\usepackage[ruled,vlined,algosection,algo2e]{algorithm2e} 

\usepackage[footnotesize]{caption}

\usepackage[colorinlistoftodos, shadow]{todonotes}

\newcounter{comment}

\newcommand{\Rbb}{\mathbb{R}}
\newcommand{\scp}[2]{\langle #1, #2 \rangle}
\newcommand{\Nbb}{\mathbb{N}}

\newtheorem{theorem}{Theorem}
\newtheorem{definition}{Definition}
\newtheorem{proposition}{Proposition}

\newtheorem{lemma}{Lemma}
\newtheorem{remark}{Remark}
\newtheorem*{example}{Example}
\newcommand{\inv}[1]{\frac{1}{#1}}

\newcommand{\tinv}[1]{{\textstyle\frac{1}{#1}}}
\newcommand{\sign}{{\rm sign}\,}
\newcommand{\Tr}{\mathrm{T}}
\newcommand{\ud}{\mathrm{d}} 
\newcommand{\binw}{\upalpha}
\newcommand{\gbinw}{\uptau}
\renewcommand{\leq}{\leqslant}
\renewcommand{\geq}{\geqslant}
\newcommand{\E}{{\mathbb{E}}}

\DeclareMathOperator{\Var}{Var}

\DeclareMathOperator{\iid}{iid}

\DeclareMathOperator{\diag}{diag}
\DeclareMathOperator{\prox}{prox}
\DeclareMathOperator{\proj}{proj}
\DeclareMathOperator{\Id}{\mathds{1}}

\DeclareMathOperator*{\argmin}{argmin}
\DeclareMathOperator*{\Argmin}{Argmin}
\newcommand{\norm}[1]{\|#1\|}
\newcommand{\lpn}[3]{\ell_{#1,#2}^{#3}}

\newcommand{\fs}{\mathsf}
\newcommand{\bs}{\boldsymbol}
\newcommand{\cl}{\mathcal}
\newcommand{\ie}{\emph{i.e.}, }
\newcommand{\eg}{\emph{e.g.}, }

\newcommand{\fnorm}[1]{|\!|\!|#1|\!|\!|}
\newcommand{\pdf}{\varphi}
\newcommand{\eqdef}{:=}
\newcommand{\defeq}{=:}

\title{Stabilizing Nonuniformly Quantized Compressed Sensing with Scalar Companders}
\author{L. Jacques, D. K. Hammond, M. J. Fadili\\[2mm]
\thanks{LJ is with the ICTEAM institute, ELEN Department, Universit\'e catholique de Louvain
(UCL), Belgium. LJ is a Postdoctoral Researcher of the Belgian National Science Foundation
(F.R.S.-FNRS).}
\thanks{DKH is with the Neuroinformatics Center, University of Oregon, USA.}
\thanks{MJF is with the GREYC, CNRS-ENSICAEN-Universit\'e de Caen, France.}
\thanks{Parts of a preliminary version of this
  work have been presented in SPARS11 Workshop (June 27-30, 2011 -
  Edinburgh, Scotland, UK), in IEEE ICIP 2011 (Sept. 11-14, 2011 -
  Brussels, Belgium) and in iTWIST Workshop (May 9-11, 2012 -
  Marseille, France).}}

\begin{document}

\maketitle

\vspace{-1cm}
\begin{abstract}
This paper addresses the problem of stably recovering sparse or compressible signals from compressed sensing measurements that have undergone optimal non-uniform scalar quantization, \ie minimizing the common $\ell_2$-norm distortion. Generally, this Quantized Compressed Sensing (QCS) problem is solved by minimizing the $\ell_1$-norm constrained by the $\ell_2$-norm distortion. In such cases, re-measurement and quantization of the reconstructed signal do not necessarily match the initial observations, showing that the whole QCS model is not \emph{consistent}. Our approach considers instead that quantization distortion more closely resembles heteroscedastic uniform noise, with variance depending on the observed quantization bin. Generalizing our previous work on uniform quantization, we show that for non-uniform quantizers described by the ``compander'' formalism, quantization distortion may be better characterized as having bounded weighted $\ell_p$-norm ($p\geq 2$), for a particular weighting. We develop a new reconstruction approach, termed Generalized Basis Pursuit DeNoise
(GBPDN), which minimizes the $\ell_1$-norm of the signal to reconstruct constrained by this weighted $\ell_p$-norm fidelity. We prove that, for standard Gaussian sensing matrices and $K$ sparse or compressible signals in $\Rbb^N$ with at least $\Omega( (K \log N/K)^{p/2})$ measurements, \ie under strongly oversampled QCS scenario, GBPDN is $\ell_2-\ell_1$ instance optimal and stable recovers all such sparse or compressible signals. The reconstruction error decreases as $O(2^{-B}/\sqrt{p+1})$ given a budget of $B$ bits per measurement. This yields a reduction by a factor $\sqrt{p+1}$ of the reconstruction error compared to the one produced by $\ell_2$-norm constrained decoders. We also propose an  primal-dual proximal splitting scheme to solve the GBPDN program which is efficient for large-scale problems. Interestingly, extensive simulations testing the GBPDN effectiveness confirm the trend predicted by the theory, that the reconstruction error can indeed be reduced by increasing $p$, but this is achieved at a much less stringent oversampling regime than the one expected by the theoretical bounds. Besides the QCS scenario, we also show that GBPDN applies straightforwardly to the related case of CS measurements corrupted by heteroscedastic Generalized Gaussian noise with provable reconstruction error reduction.
\end{abstract}

\section{Introduction}
\label{sec:introduction}

\subsection{Problem statement}

Measurement quantization is a critical step in the design and in
the dissemination of new technologies implementing the Compressed Sensing
(CS) paradigm. Quantization is indeed mandatory for transmitting,
storing and even processing any data sensed by a CS device.  

In its most popular version, CS provides uniform theoretical guarantees for stably recovering
any sparse (or compressible) signal at a sensing rate
proportional to the signal intrinsic dimension (\ie its \emph{sparsity} level)
\cite{donoho2006cs,candes2008rip}. However, the distortion
introduced by any quantization step is often still crudely modeled as
a noise with bounded $\ell_2$-norm.

Such an approach results in reconstruction methods aiming at finding a sparse signal estimate for which the sensing is close, in a $\ell_2$-sense, to the available quantized signal observations. However, earlier works have pointed out that this method is not optimal. For instance, \cite{goyal1998_quantovercomp} analyses the error achieved when a signal is reconstructed from its quantized coefficients in some overcomplete expansion. Translated to our context, this amounts to the ideal CS scenario where some \emph{oracle} provides us the true signal support knowledge. In this context, a linear \emph{least square} (LS) reconstruction minimizing the $\ell_2$-distance in the coefficient domain is inconsistent and has a \emph{mean square error} (MSE) decaying, at best, as the inverse of the frame redundancy factor. Interestingly, any \emph{consistent} reconstruction method, \ie for which the quantized coefficients of the reconstructed signal match those of the original signal, shows a much better behavior since its MSE is in general lower-bounded by the inverse of the \emph{squared} frame redundancy; this lower bound being attained for specific overcomplete Fourier frames.

A few other works in the Compressed Sensing literature have also considered the quantization distortion differently. In \cite{Dai2009}, an adaptation of both Basis Pursuit DeNoise (BPDN) program and the Subspace Pursuit algorithm integrates an explicit constraint enforcing consistency. In \cite{Zymnis2009}, nonuniform quantization noise and Gaussian noise in the measurements before quantization are properly dealt with using an $\ell_1$-penalized maximum likelihood decoder.

Finally, in \cite{LasBouDav::2009::Demcracy-in-action,Jacques2011,Plan2011}, the
  extreme case of 1-bit CS is studied, \ie when only the signs of the measurements
are sent to the decoder. These works have shown that consistency with
the 1-bit quantized measurements is of paramount importance for
reconstructing the signal where straightforward methods relying on $\ell_2$
fidelity constraints reach poor estimate quality.

\subsection{Contributions}

The present work addresses the problem of recovering sparse or compressive signals in a \emph{given} non-uniform Quantized Compressed Sensing (QCS) scenario. In particular, we assume that the signal measurements have undergone an optimal non-uniform scalar quantization process, \ie optimized a priori according to a common minimal distortion standpoint with respect to a source with known probability density function (pdf). This \emph{post-quantization} reconstruction strategy, where only increasing the number of measurements can improve the signal reconstruction,
is inspired by other works targeting consistent reconstruction approaches in comparison with methods advocating solutions of minimal $\ell_2$-distortion \cite{goyal1998_quantovercomp,Dai2009,Jacques2010}.
Our work is therefore distinct from approaches where other
quantization schemes (\eg $\Sigma\Delta$-quantization \cite{gunturk2010sobolev})
are tuned to the global CS formalism or to specific CS decoding
schemes (\eg Message Passing Reconstruction
\cite{kamilov2011optimal}). These techniques often lead to signal reconstruction
MSE rapidly decaying with the measurement number $M$~--~for instance, a
$r$-order $\Sigma\Delta$-quantization of CS measurements combined with
a particular reconstruction procedure has a MSE decaying nearly as
$O\big(M^{-r+\inv{2}}\big)$ \cite{gunturk2010sobolev}~--~but their application
involves generally more involved quantization strategies at the CS encoding stage.

This paper also generalizes the results provided in \cite{Jacques2010}
to cover the case of non-uniform scalar quantization of CS
measurements. We show that the theory of ``Companders''
\cite{gray1998quantization} provides an elegant framework for
stabilizing the reconstruction of a sparse (or compressible) signal
from non-uniformly quantized CS measurements.  Under the \emph{High Resolution Assumption} (HRA), \ie
when the bit budget of the quantizer is high and the quantization bins
are narrow, the compander theory provides an equivalent description of the action of a quantizer through
sequential application of a \emph{comp}ressor, a uniform quantization,
then an ex\emph{pander} (see Section~\ref{sec:quant-fram} for
details). As will be clearer later, this equivalence allows us to define new distortion
constraints for the signal reconstruction which are more faithful
to the non-uniform quantization process given a certain  QCS measurement regime.

Algorithms for reconstructing from quantized measurements commonly
rely on mathematically describing the noise induced by quantization as
bounded in some particular norm. A data fidelity constraint reflecting this fact is then incorporated in the reconstruction method. Two natural examples of such
constraints are that the $\ell_2$-norm be bounded, or that the
quantization error be such that the unquantized values lie in
specified, known quantization bins. In this paper, guided by the
compander theory, we show that these
two constraints can be viewed as special (extreme) cases of a particular
{\emph{weighted}} $\ell_p$-norm, which forms the basis for our
reconstruction method. The weights are determined from a set of $p$-optimal quantizer levels,
 that are computed from the observed quantized values. We draw the
 reader attention to the fact these weights do not depend on the
 original signal which is of course unknown. They are used only for
 signal reconstruction purposes, and are optimized with respect to the
 weighted norm. In the QCS framework, and owing to the
particular weighting of the norm, each quantization bin contributes
equally to the related global distortion.

Thanks to a new estimator of the weighted $\ell_p$-norm of the
quantization distortion associated to these particular levels (see
Lemma~\ref{lemma:bound-lpw-gaussian-vector}), and with the proviso that the
sensing matrix obeys a generalized Restricted Isometry Property (RIP)
expressed in the same norm (see \eqref{eq:rip-p}), we show that
solving a General Basis Pursuit DeNoising program (GBPDN) -- an
$\ell_1$-minimization problem constrained by a weighted $\ell_p$-norm whose radius
is appropriately estimated -- stably recovers strictly sparse or
compressible signals (see Theorem~\ref{prop:l2-l1-instance-optimality-GBPDN}).

We also quantify precisely the reconstruction error of GBPDN as a
function of the quantizer bit rate (under the HRA) for any value of
$p$ in the weighted $\ell_p$ constraint. These results reveal a set of
conflicting considerations for setting the optimal $p$. On the
  one hand, given a budget of $B$ bits per measurement and for a high
  number of measurements $M$, the error decays as
  $O(2^{-B}/\sqrt{p+1})$ when $p$ increases (see
  Proposition~\ref{prop:towards-quant-cons}), \ie a favorable
  situation since then GBPDN tends also to a consistent
  reconstruction method. On the other hand, the
larger $p$, the greater the number of measurements required to ensure
that the generalized RIP is fulfilled. In particular, one needs
$\Omega((K \log N/K)^{p/2})$ measurements compared to a $\ell_2$-based CS bound of
$\Omega(K \log N/K)$ measurements (see
Proposition~\ref{prop:grip-gauss}). Put differently, given a certain number
of measurements, the range of theoretically admissible $p$ is upper bounded, an effect which is expected since the error due to quantization cannot be eliminated in the reconstruction.

In fact, the stability of GBPDN in the context of QCS is a consequence of a an even more general stability result that holds for a broader class additive heteroscedastic measurement noise having a bounded weighted $\ell_p$ norm. This for instance covers the case of heteroscedastic Generalized Gaussian noise where the constraint of GBPDN  can be interpreted as a (variance) stabilization of the measurement distortion, see Section~\ref{sec:case-heter-nois}).
  
\subsection{Relation to prior work}

Our work is novel in several respects. For instance, as
stated above, the quantization distortion in the literature is often
modeled as a mere Gaussian noise with bounded variance
\cite{Dai2009}. In \cite{Jacques2010}, only uniform quantization is
handled and theoretically investigated. In \cite{Zymnis2009},
nonuniform quantization noise and Gaussian noise are handled but
theoretical guarantees are lacking. To the best of our knowledge, this
is the first work thoroughly investigating the theoretical guarantees
of $\ell_1$ sparse recovery from non-uniformly quantized CS
measurements, by introducing a new class of convex $\ell_1$
decoders. The way we bring the compander theory in the picture to
compute the optimal weights from the quantized measurements is also an
additional originality of this work. 

\subsection{Paper organization}
\label{sec:paper-organization}

The paper is organized as follows. In Section~\ref{sec:non-uniform-scalar}, we recall the
theory of optimal scalar quantization seen through the compander
formalism. We then explain how this point of view can help us in
understanding the intrinsic constraints that quantized CS measurements
must satisfy, and we introduce a new distortion measure, the
$p$-Distortion Consistency, expressed in terms of
a weighted $\ell_p$-norm. Section~\ref{sec:weight-fidel-guar}
introduces the GBPDN CS class of decoders integrating 
weighted $\ell_p$-constraints, and describes sufficient conditions for guaranteeing
reconstruction stability. This section shows also the generality of
this procedure for stabilizing additive heteroscedastic GGD measurement noise
during the signal reconstruction. In Section~\ref{sec:dequ-with-gener}, we
explain how GBPDN can be used for reconstructing a
signal in QCS when its fidelity constraint is adjusted to the parameters
defined in Section~\ref{sec:p-optimal-levels}. We show
that this specific choice leads to a (variance) stabilization of the quantization distortion forcing each
quantization bin to contribute equally to the overall distortion error.
In Section~\ref{sec:numer-exper}, we describe a provably convergent
primal-dual proximal splitting algorithm to solve the GBPDN program, and demonstrate the power of the proposed approach with several numerical experiments on sparse signals.

\subsection{Notation} 

All finite space dimensions are denoted by capital letters (\eg
$K,M,N,D\in \Nbb$), vectors (resp. matrices) are written in small
(resp. capital) bold symbols. For any vector $\bs u$, the
$\ell_p$-norm for $1\leq p < \infty$ is $\|\bs u\|_p=\left(\sum_{i}
  |u_i|^p\right)^{1/p}$, as usual $\|\bs u\|_\infty=\max_i |u_i|$ and
we write $\|\bs u\|=\|\bs u\|_2$. We write $\|\bs u\|_0 =\#\{i:u_i\neq
0\}$, which counts the number of non-zero components.  We denote the
set of $K$-sparse vectors in the canonical basis by $\Sigma_K=\{\bs
u\in\Rbb^N:\|\bs u\|_0 \leq K\}$. When necessary, we write $\ell^D_p$
as the normed vector space $(\Rbb^D,{\|\!\cdot\!\|_p})$.

The identity matrix in $\Rbb^D$ is written $\Id_D$ (or simply $\Id$ if
the $D$ is clear from the context). $\bs U = \diag(\bs u)$ is the
diagonal matrix with diagonal entries from $\bs u$, \ie $\bs
U_{ij}=u_i\delta_{ij}$. Given the $N$-dimensional signal space
$\Rbb^N$, the index set is $[N]=\{1,\,\cdots,N\}$, and $\bs
\Phi_I\in\Rbb^{M\times \# I}$ is the restriction of the columns of
$\bs \Phi$ to those indexed in the subset $I\subset [N]$, whose
cardinality is $\#I$.  
Given $\bs x\in\Rbb^N$, $\bs x^{\bs \Psi}_{K}$ stands for the best
$K$-term $\ell_2$-approximation of $\bs x$ in the orthonormal basis
$\bs \Psi\in\Rbb^{N\times N}$, that is, $\bs x^{\bs \Psi}_K =
\bs\Psi\big(\argmin \{ \|\bs x - \bs \Psi \bs \zeta\|:
\bs\zeta\in\Rbb^N, \|\bs \zeta\|_0\leq K \}\big)$. When $\bs \Psi =
\Id$, we write $\bs x_K=\bs x^{\Id}_K$ with $\|\bs x_K\|_0\leq K$.  A
random matrix $\bs \Phi \sim \cl N^{M\times N}(0,1)$ is a $M\times N$
matrix with entries $\bs \Phi_{ij} \sim_{\iid} \cl N(0,1)$. The 1-D
Gaussian pdf of mean $\mu \in \Rbb$ and variance $\sigma^2 \in
\Rbb^*_+$ is denoted $\gamma_{\mu,\sigma}(t) \eqdef
(2\pi\sigma^2)^{-1/2}\,\exp(-\tfrac{(t-\mu)^2}{2\sigma^2})$.

For a function $f:\Rbb\to\Rbb$, we write $\fnorm{f}_q\eqdef(\int_{\Rbb}\,\ud t\, |f(t)|^q)^{1/q}$, with
$\fnorm{f}_\infty\eqdef\sup_{t \in \Rbb} |f(t)|$. 

In order to state many results which hold asymptotically as a
dimension $D\in\Rbb$ increases, we will use the common Landau
family of notations, \ie the symbols $O$, $\Omega$, $\Theta$, $o$,
and $\omega$ (their exact definition can be found in
\cite{knuth1976big}).
Additionally, for $f,g\in C^1(\Rbb_+)$, we write $f(D)\simeq_D g(D)$ when $f(D) = g(D)(1+
o(1))$. We also introduce two new asymmetric notations dealing with asymptotic
quantity ordering,~\ie
\begin{align*}
 f(D)&\lesssim_D g(D)\hspace{-2.5cm}&&\Leftrightarrow\ \exists\,\delta:\Rbb\to\Rbb_+:\ f(D) +
 \delta(D) \simeq_D g(D)\\
 f(D)&\gtrsim_D g(D)\hspace{-2.5cm}&&\Leftrightarrow\ -f(D) \lesssim_D
 -g(D) .
\end{align*}
If any of the asymptotic relations above hold with respect
to several large dimensions $D_1,D_2,\,\cdots$, we write
$\simeq_{D_1,D_2,\,\cdots}$ and correspondingly for $\lesssim$ and
$\gtrsim$.

\section{Non-Uniform Quantization in Compressed Sensing}
\label{sec:non-uniform-scalar}

Let us consider a signal $\bs x \in \Rbb^N$ to be measured. We assume
that it is either strictly sparse or compressible, in a prescribed orthonormal
basis $\bs \Psi=\big(\bs\Psi_1,\,\cdots,\bs\Psi_N)\in\Rbb^{N\times
  N}$. This means that the signal $\bs x = \bs \Psi
\bs \zeta = \sum_j \bs \Psi_j \zeta_j$ is such that the
$\ell^N_2$-approximation error $\|\bs \zeta - \bs \zeta_K\|=\|\bs x - \bs x^{\bs \Psi}_K\|$ quickly decreases
(or vanishes) as $K$ increases. For the sake
of simplicity, and without loss of generality, the sparsity basis is taken in the sequel as the standard basis, \ie
$\bs\Psi = \Id$, and $\bs \zeta$ is identified with $\bs
x$. All the results can be readily extended to other orthonormal bases $\bs \Psi
\neq \Id$.

In this paper, we are interested in compressively sensing $\bs
x\in\Rbb^N$ with a given measurement matrix $\bs \Phi\in\Rbb^{M\times
  N}$. Each CS measurement, \ie each entry of $\bs z = \bs \Phi\bs x$, 
  undergoes a general scalar \emph{quantization}. We will assume this quantization to be optimal
relative to a known distribution of each entry $z_i$.  For
simplicity, we only consider matrices $\bs \Phi$ that yield $z_i$ to
be i.i.d. $\cl N(0,\sigma_0^2)$ Gaussian, with pdf
$\varphi_0\eqdef\gamma_{0,\sigma_0}$. This is satisfied, for instance,
if $\bs \Phi \sim \cl N^{M\times N}(0,1)$, with $\sigma_0 = \|\bs
x\|_2$. When $\bs \Phi=[\bs \varphi_1^T,\,\cdots,\bs \varphi_M^T]^T$ is a (fixed) realization of $\cl N^{M\times N}(0,1)$, the entries $z_j = \scp{\bs\varphi_j}{\bs x}$ of the vector $\bs z = \bs \Phi \bs x$ are $M$ (fixed) realizations of the same Gaussian distribution $\cl N(0, \|\bs x\|^2)$. It is therefore legitimate to quantize these values optimally using the normality of the source.\footnote{Avoiding pathological situations where $\bs x$ is adversarially forged knowing $\bs \Phi$ for breaking this assumption.}.

Our quantization scenario uses a $B$-bit quantizer $\cl Q$ which has
been optimized with respect to the measurement pdf $\varphi_0$ for
$\cl B=2^B=\#\Omega$ levels $\Omega=\{\omega_k:1\leq k\leq \cl B\}$
and thresholds $\{t_k:1\leq k\leq \cl B+1\}$ with $-t_1=t_{\cl B+1} =
+\infty$. Unlike the framework developed in
\cite{Zymnis2009}, our sensing scenario considers that any noise
corrupting the measurements before quantization is negligible
compared to the quantization distortion.

Consequently, given a measurement matrix $\bs \Phi\in\Rbb^{M\times N}$, our
quantized sensing model is
\begin{equation}
  \label{eq:quant-sensing-model}
  \bs y\ =\ \cl Q[\bs \Phi \bs x]\ =\ \cl Q[\bs z]\ \in \Omega^M. 
\end{equation}

Following recent studies \cite{Dai2009,laska2009finite,Jacques2010} in
the CS literature, this work is interested in optimizing the signal
reconstruction stability from $\bs y$ under different sensing
conditions, for instance, when the \emph{oversampling ratio} $M/K$ is
allowed to be large. Before going further into this signal sensing
model, let us describe first the selected quantization framework. The latter
is based on a scalar quantization of each component of the signal
measurement vector.

\subsection{Quantization, Companders and Distortion}
\label{sec:quant-fram}

A scalar quantizer $\cl Q$ is defined from $\cl B=2^B$
\emph{levels} $\omega_k$ (coded by $B=\log_2 \cl B$ bits) and
$\cl B+1$ \emph{thresholds} $t_k\in \Rbb\,\cup\, \{\pm \infty\}
= \overline\Rbb$, with $\omega_{k}<\omega_{k+1}$ and $t_{k}\leq
\omega_{k} < t_{k+1}$ for all $1\leq k \leq \cl B$.
The $k^{\rm th}$ quantizer \emph{bin} (or \emph{region})
is $\cl R_k=[t_{k},t_{k+1})$, with bin width $\gbinw_k=t_{k+1}-t_k$.
The quantizer $\cl Q$ is a map: $\Rbb \to \Omega = \{\omega_k: 1\leq k \leq \cl B\}$, $t \mapsto \cl Q[t] = \omega_k \iff
t\in \cl R_k$.
An optimal scalar quantizer $\cl Q$ with respect to a random source
$\cl Z$ with pdf $\pdf_{\cl Z}$ is such that the distortion
$\E |\cl Z -\cl Q[\cl Z]|^2$ is minimized. Optimal levels and
thresholds can be calculated for a fixed number of quantization bins
by the Lloyd-Max Algorithm \cite{Lloyd1982,Max1960}, or by an
asymptotic (with respect to~$B$) \emph{companding} approach
\cite{gray1998quantization}.

Throughout this paper, we work under the HRA. This means that, given the source
pdf $\pdf_{\cl Z}$, the number of bits $B$ is sufficient to
validate the approximation
\begin{equation*}
  \pdf_{\cl Z}(t)
  \simeq_B \pdf_{\cl Z}(\omega_k),\quad \forall t\in \cl R_k. \eqno{({\textbf{HRA})}}.  
\end{equation*}

A common argument in quantization theory \cite{gray1998quantization}
states that under the HRA, every optimal regular quantizer 
can be described by a compander (a
portemanteau for ``\textbf{comp}ressor'' and
``exp\textbf{ander}''). More precisely, we have
$$
\cl Q = \cl G^{-1}\circ\cl Q_{\binw}\circ \cl G,
$$  
with $\cl G:\Rbb\to[0,1]$ a bijective function called the
\emph{compressor}, $\cl Q_{\binw}$ a uniform quantizer of the
interval $[0,1]$ of bin width $\binw=2^{-B}$, and the inverse mapping
$\cl G^{-1}:[0,1]\to\Rbb$ called the \emph{expander}.

For optimal quantizers the compressor $\cl G$ maps the thresholds
$\{t_k: 1\leq k\leq \cl B\}$ and the levels $\{\omega_k\}$ into the values
\begin{equation}
  \label{eq:thresh-def-with-compressor}
  t'_k\eqdef\cl G(t_k)=(k-1)\binw,\qquad \omega'_k\eqdef\cl
  G(\omega_k)=(k-1/2)\binw ,
\end{equation}
and under the HRA the optimal $\cl G$ satisfies
\begin{equation}
\label{eq:optimal-compressor}
  \cl G':=\tfrac{\ud}{\ud \lambda} \cl G(\lambda)\ =\ \bigg[\int_{\Rbb} \pdf_{\cl Z}^{1/3}(t)\, \ud t\,\bigg]^{-1}\,\pdf_{\cl Z}^{1/3}(\lambda).
\end{equation}
Intuitively, the function $\cl G'$, also called \emph{quantizer point
    density function} (qpdf) \cite{gray1998quantization}, relates the
  quantizer bin widths before and after domain compression by $\cl
  G$. Indeed, under HRA, we can show that $\cl G'(\lambda) \simeq
  \upalpha/\tau_k$ if $\lambda \in \cl R_k$. We will see later that
  this function is the key to conveniently weight some new quantizer
  distortion measures.

We note that, for $\pdf_{\cl Z}(t) = \gamma_{0,\sigma}(t)$ with 
cumulative distribution function $\phi_{\cl
  Z}(\lambda;\sigma^2)
= \inv{2}{\rm erfc}(-\tfrac{\lambda}{2\sigma})$ so that $\phi_{\cl
  Z}^{-1}(\lambda';\sigma^2)=\sigma\sqrt{2}\,{\rm
  erf}^{-1}(2\lambda'-1)$, we have $\cl G(\lambda) = \phi_{\cl
  Z}(\lambda;3\sigma^2)$ and $\cl G^{-1}(\lambda')=\phi_{\cl
  Z}^{-1}(\lambda';3\sigma^2)$.

The application of $\cl G$ modifies the source $\cl Z$ such that $\cl
G(\cl Z)-\cl G(\cl Q[\cl Z])$ behaves more like a uniformly
distributed random variable over $[-\binw/2,\binw/2]$. 
The compander formalism predicts the distortion of optimal
scalar quantizer under HRA. For high bit rate $B$, the Panter and Dite formula
\cite{panter1951quantization} states that
\begin{equation}
\label{eq:panter-dite-eqt}
\E |\cl Z-\cl Q[\cl Z]|^2\ \mathop{\simeq}_B\ \tfrac{2^{-2B}}{12}\,\int_{\Rbb} \cl G'(t)^{-2}\,\pdf_{\cl Z}(t)\,\ud
t\ =\ \tfrac{2^{-2B}}{12}\,\bigg(\int_{\Rbb} \,\pdf_{\cl Z}^{1/3}(t)\,\ud t\bigg)^3\ =\
\tfrac{2^{-2B}}{12}\,\fnorm{\pdf_{\cl Z}}_{1/3}.
\end{equation}

Finally, we note that by the construction defined in
\eqref{eq:thresh-def-with-compressor}, the quantized values $\cl
Q[\lambda]$ satisfy
\begin{equation}
\label{eq:compressor-consistent-rel}
|\cl G(\lambda)-\cl G(\cl Q[\lambda])|\leq \binw/2,\quad\forall \lambda\in\Rbb.
\end{equation}
We describe in the next sections how
\eqref{eq:compressor-consistent-rel} and \eqref{eq:panter-dite-eqt}
may be viewed as two extreme cases of a general class of constraints
satisfied by a quantized source $\cl Z$.

\subsection{Distortion and Quantization Consistency}
\label{sec:dist-quant-quant-cons}

Let us consider the sensing model \eqref{eq:quant-sensing-model}, for
which the scalar quantizer $\cl Q$ and associated compressor $\cl G$
are optimal relative to the measurements $\bs z = \bs \Phi \bs x$ whose entries $z_i$ are $\iid$ realizations of $\cl
N(0,\sigma_0^2)$.
In the compressor domain we may write
$$
\cl G(\bs y) \ =\ \cl G(\bs z) + (\cl G(\cl Q[\bs z]) - \cl G(\bs z)) =
\cl G(\bs z) + \bs \varepsilon,
$$
where $\bs \varepsilon$ represents the quantization
distortion. \eqref{eq:compressor-consistent-rel} then shows
that 
$$
\|\bs \varepsilon\|_{\infty} = \|\cl G(\cl Q[\bs z]) - \cl
G(\bs z)\|_{\infty} \leq \binw/2.
$$

Naively, one may expect any reasonable estimate $\bs x^*$ of $\bs x$
(obtained by some reconstruction method) to reproduce the same
quantized measurements as originally observed. Inspired by the
  terminology introduced in \cite{thao1994msereducconsist,goyal1998_quantovercomp},
we say that $\bs x^*$
satisfies the \emph{quantization consistency} (QC) if $\cl Q[\bs \Phi
\bs x^*] = \bs y$. From the previous reasoning this is equivalent to 
\begin{equation*}
\|\cl G(\bs \Phi \bs x^*) - \cl G(\bs y)\|_{\infty}\ \leq\ \epsilon_{\rm
QC} \eqdef \binw/2. \eqno{({\bf QC})}  
\end{equation*}

At first glance, it is tempting to try to impose directly QC in the data fidelity constraint.
However, as will be revealed by our analysis, directly imposing QC does
\emph{not} lead to an effective QCS reconstruction algorithm. This
counterintuitive effect, already observed in the case of signal
recovery from uniformly quantized CS
\cite{Jacques2010}, is due to the specific requirements that the
sensing matrix should respect to make such a consistent reconstruction method
stable.  

In contrast the Basis Pursuit DeNoise (BPDN) program \cite{Chen98atomic} enforces a
constraint on the $\ell_2$ norm of the reconstruction quantization
error, which we will call distortion consistency.
For BPDN, the estimate $\bs x^*$ is provided by
$$
\bs x^*\ \in\ \Argmin_{\bs u \in \Rbb^N}\,\|\bs u\|_1\ {\rm s.t.}\ \|\bs y -\bs
\Phi\bs u\| \leq\epsilon_{\rm DC},
$$
where the bound $\epsilon^2_{\rm DC} :=
M\,\tfrac{\sqrt{3}\,\pi}{2}\,\sigma_0^2\, 2^{-2B}$ is dictated by
 the Panter-Dite formula. According to the Strong Law of Large
Numbers (SLLN) obeyed by the HRA, and since $z_i$ are $\iid$ realizations of $Z \sim
\cl N(0,\sigma_0^2)$, the following holds almost surely
\begin{equation}
\label{eq:panter-dite}
\tinv{M} \|\bs z - \cl Q[\bs z]\|^2\ \mathop{\simeq}_{M}\ \E|\cl Z -
\cl Q[\cl Z]|^2\ \mathop{\simeq}_{B}\ \tfrac{2^{-2B}}{12} \fnorm{\pdf_0}_{1/3}
= \  \tfrac{\sqrt{3}\,\pi}{2}\,\sigma_0^2\, 2^{-2B} .   
\end{equation}
Accordingly, we say that any estimate $\bs x^*$ satisfies \emph{distortion consistency} (DC) if
\begin{equation*}
\|\bs \Phi \bs x^* - \bs y\| \leq \epsilon_{\rm DC}.
\eqno{({\bf DC})}
\end{equation*}

However, as stated for the uniform quantization case in
\cite{Jacques2010}, DC and QC do not imply each other.  In
particular, the output $\bs x^*$ of BPDN needs not satisfy
quantization consistency. A major motivation for the present work is
the desire to develop provably stable QCS recovery methods based on measures of
quantization distortion that are as close as possible to QC.

\subsection{$p$-Distortion Consistency}
\label{sec:p-optimal-levels}

This section shows that the QC and DC constraints may be seen as limit cases
of a weighted $\ell_p$-norm description of the quantization
distortion.
The expression of the appropriate weights in the weighted $\ell_p$ norm
will depend both on the $p$-optimal quantizer levels, described below, and of the quantizer
point density function $\cl G'$ introduced in Section~\ref{sec:quant-fram}.

For the Gaussian pdf $\pdf_0=\gamma_{0,\sigma_0}$, given a set of
thresholds $\{t_k: 1\leq k\leq \cl B\}$, we define the $p$-\emph{optimal quantizer
  levels} $\omega_{k,p} \in \overline\Rbb$ as
\begin{equation} \label{eq:wkp_integral}
\omega_{k,p} \eqdef \argmin_{\lambda\in\cl R_k} \int_{\cl R_k} |t - \lambda|^p\ \pdf_0(t)\,\ud t,
\end{equation}
for $2\leq p<\infty$, and $\omega_{k,\infty} \eqdef
\tinv{2}(t_k+t_{k+1})$. These generalized levels were for
  instance already defined by Max in
  his minimal distortion study \cite{Max1960}, and their definition
  \eqref{eq:wkp_integral} is also
  related to the concept of minimal $p^{\rm th}$-power distortion
  \cite{gray1998quantization}.  For $p=2$, we find the definition of the
initial quantizer levels, \ie $\omega_{k,2}=\omega_k$. In this paper,
we always assume that $p$ is a positive integer but all our analysis can be extended
to the positive real case. As proved in Appendix~\ref{sec:p-optimal-level-definiteness}, the $p$-optimal levels
are well-defined.

\begin{lemma}[\bf $p$-optimal Level Well-Definiteness]
\label{lemma:p-optimal-levels}
  The $p$-optimal levels $\omega_{k,p}$ are uniquely defined.
  Moreover, for $\sigma_0>0$, $\lim_{p\to +\infty} \omega_{k,p} = \omega_{k,\infty}$,
with $|\omega_{k,p}|=\Omega(\sqrt p)$ for $k\in\{1,\cl B\}$.
\end{lemma}

Using these new levels, we define the (suboptimal) quantizers $\cl Q_p$ (with $\cl Q_2=\cl Q$)
such that
\begin{equation}
\label{eq:Qp}
\cl Q_p[t] = \omega_{k,p}\ \Leftrightarrow\ t\in\cl R_k = \cl
Q_p^{-1}[\omega_{k,p}] = \cl Q^{-1}[\omega_k].
\end{equation}

Two important points must be explained regarding the definition of
$\cl Q_p$.  First, the \emph{(re)quantization} of any source $\cl Z$ with
$\cl Q_p$ is possible from the knowledge of the quantized
value $\cl Q[\cl Z]$, as $\cl Q_p[\cl Z]=\cl Q_p[\cl Q[\cl Z]]$
since both quantizers share the same decision thresholds. Second,
despite the sub-optimality of $\cl Q_p$ relative to the untouched
thresholds $\{t_k: 1\leq k\leq \cl B\}$, 
we will see later that introducing this quantizer provides improvement
in the modeling of $\cl Q_p[\cl Z] - \cl Z$ by a Generalized Gaussian
Distribution (GGD) in each quantization bin.

\begin{remark}
Unfortunately, there is no closed form formula for computing
$\omega_{k,p}$. However, as detailed in
Appendix~\ref{sec:comp-omeg-p}, they can be computed up to numerical
precision using Newton's method combined with simple numerical
quadrature for the integral in \eqref{eq:wkp_integral}.
\end{remark}

Given $p\geq 2$ and for high $B$, the asymptotic behavior of a quantizer $\cl Q_p$ and
of its $p^{\rm th}$ \emph{power distortion} $\int_{\cl R_k} |t -
\omega_{k,p}|^p\ \pdf_0(t)\,\ud t$ in each bin $\cl R_k$ follows two
very different regimes in $\Rbb$ governed by a particular transition value $T =
\Theta(\sqrt B)$.
This is described in the following lemma (proved in Appendix~\ref{proof-toolbox-lemma}), which, to the best of
our knowledge, provides new results and may
be of independent interest for characterizing Gaussian source
quantization (even for the standard case $p=2$).

\begin{lemma}[\bf Asymptotic $p$-Quantization Characterization]
\label{lemma:useful-ineq}
Given the Gaussian pdf $\pdf_0$ and its associated compressor $\cl G$
function, choose $0<\beta<1$ and $p\in\Nbb$, and define the transition value
$$
T=T(B)=(6\,\sigma_0^2(\log 2^\beta)\,B)^{1/2}.
$$  
$T$ defines two specific asymptotic regimes for the
quantizer $\cl Q_p$:
\begin{enumerate}
\item The \emph{vanishing bin regime} $\cl T
= [-T, T]$: for all $\cl R_k \subset \cl T$ and any $c\in
\cl R_k$, the bin widths decay as $\gbinw_k =
O(2^{-(1-\beta)B})$, and the the related $p^{\rm
  th}$-power distortion and \emph{qpdf} asymptotically obey
\begin{align}
\label{lemma:useful-ineq-eq7}
\textstyle \int_{\cl R_k} |t-\omega_{k,p}|^p\ \pdf_0(t)\,\ud t\ &\simeq_B\
\tfrac{\gbinw_k^{p+1}}{(p+1)\,2^{p}}\,\pdf_0(c),\\
\label{lemma:useful-ineq-eq8}
\cl G'(c)\ &\simeq_B \tfrac{\binw}{\gbinw_k}.
\end{align}
\item The \emph{vanishing distortion regime} $\cl T^c$: we have 
$\cl G'(t) \leq \cl G'(T(B)) = \Theta(2^{-\beta B})$ for all $t \in \cl
T^c$. Moreover, the number of bins in $\cl T^c$ and their $p^{\rm
  th}$-power distortion decay, respectively, as
\begin{align}
\label{lemma:useful-ineq-eq2}
\#\{k : \cl R_k \subset \cl T^c\}&=
\Theta\big(B^{-1/2}\,2^{(1-\beta)B}\big),\\
\label{lemma:useful-ineq-eq2bis}
\int_{\cl R_k} |t-\omega_{k,p}|^p\
\pdf_0(t)\,\ud t&= O\big(B^{-(p+1)/2}\,2^{-3\beta B}\big),\quad \forall 
\cl R_k \subset \cl T^c.
\end{align}
\end{enumerate}
\end{lemma}

We now state an important result, proved in Appendix
\ref{proof-lemma-bound-lpw-gaussian-vector} from the statements of
Lemma~\ref{lemma:useful-ineq}, which, together with the SLLN, estimates the
quantization distortion of $\cl Q_p$ on a random Gaussian vector. Given $p\geq 1$ and some positive weights $\bs w=(w_1,\cdots,w_M)^T\in\Rbb^M_+$, this distortion is measured by a
weighted $\ell_p$-norm defined as\footnote{A more
  standard weighted $\ell_p$-norm definition reads $(\sum_i w_i
  |v_i|^p)^{1/p}$. Our definition choice, which is strictly equivalent,
  offers useful writing simplifications, \eg when observing
  that $\|\bs\Phi \bs x\|_{p,\bs w}=\|\bs\Phi' \bs x\|_{p}$ with $\bs \Phi'=\diag(\bs w)\bs\Phi$.} $\|\bs v\|_{p,\bs w}\eqdef\|\diag(\bs w)\,\bs v\|_p$ for any $\bs
v\in\Rbb^M$.

\begin{lemma}[\bf Asymptotic Weighted $\ell_{p}$-Distortion]
  \label{lemma:bound-lpw-gaussian-vector}
  Let $\bs z\in\Rbb^M$ be a random vector where each component $z_i \sim_{\iid}
  \pdf_0$. Given the optimal compressor function $\cl G$
  associated to $\pdf_0$ and the weights $\bs w=\bs w(p)$ such that
  $w_i(p) = \cl G'\big(\cl Q_p[z_i]\big)^{(p-2)/p}$ for $p\geq 2$, the following holds almost surely
  \begin{equation}
    \label{eq:eq_lemma_Qp_dist_bound}
    \|\cl Q_p[\bs z] - \bs z\|^p_{p,\bs w}\ \mathop{\simeq}_{B,M}\ 
    M\, \tfrac{2^{-Bp}}{(p+1)\,2^{p}}\,\fnorm{\pdf_0}_{1/3}\ \defeq\ \epsilon_p^p,    
  \end{equation}
with $\fnorm{\pdf_0}_{1/3}=2\pi\,\sigma_0^2\,3^{3/2}$.
\end{lemma}

This lemma provides a tight estimation for $p=2$ and
$p\to+\infty$. Indeed, in the first case $\bs w=\bs 1$ and the bound
matches the Panter-Dite estimation \eqref{eq:panter-dite}. For
$p\to\infty$, we observe that $\epsilon_\infty = 2^{-(B+1)} = \binw/2
= \epsilon_{\rm QC}$.

\begin{figure}
  \centering
  \includegraphics[width=7.5cm]{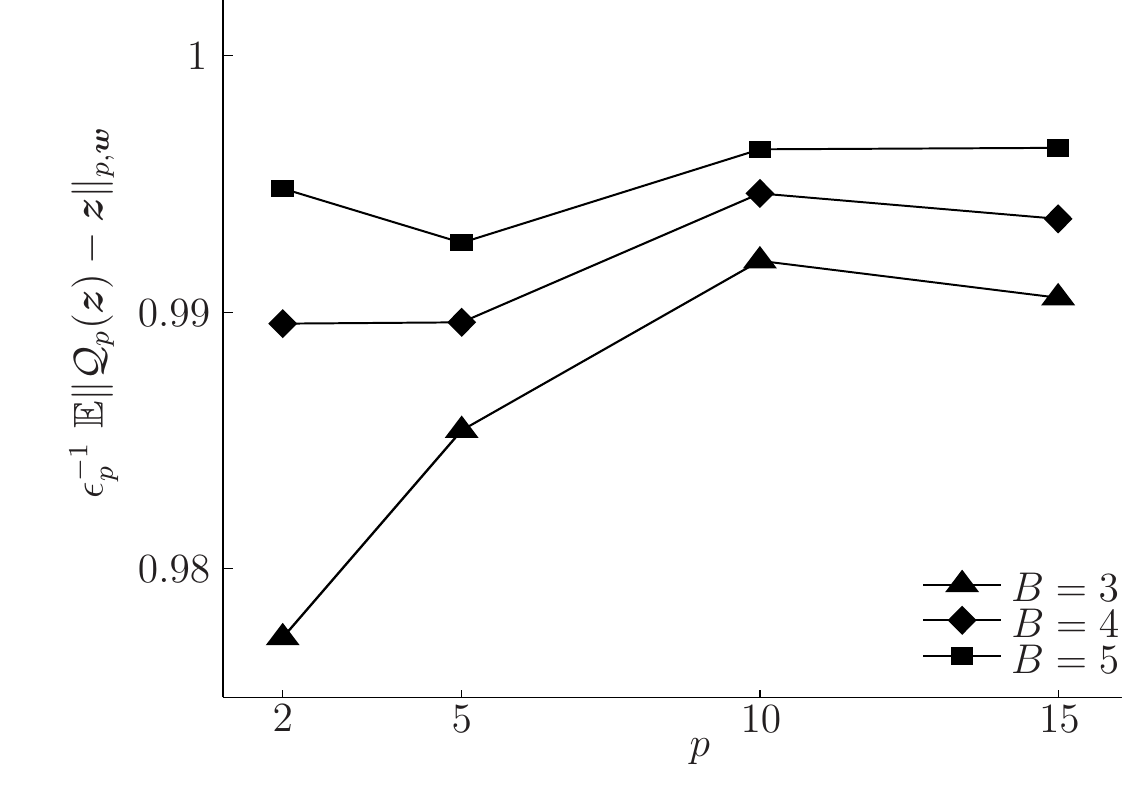}
  \caption{\label{fig:testing_eps_p_bound}
Comparing the theoretical bound $\epsilon_p$ to the empirical mean estimate of $\E \|\cl
  Q_p[\bs z] - \bs z\|_{p,\bs w}$ using 1000 trials of Monte-Carlo simulations, for each
  $B=3,4,5$).}
\end{figure}

Fig.~\ref{fig:testing_eps_p_bound} shows how well the $\epsilon_p$
estimates the distortion $\|\cl Q_p[\bs z] - \bs z\|_{p,\bs w}$ for
the weights and the $p$-optimal levels given in Lemma 2. This has been
measured by averaging this quantization distortion for 1000 realizations of a
Gaussian random vector $\sim \cl N^M(0,1)$ with $M=2^{10}$, $p \in
\{2, \,\cdots, 15\}$ and $B=3,4$ and 5.  We observe that the bias of
$\epsilon_p$, as reflected here by the ratio $\epsilon_p^{-1}\,\E\|\cl
Q_p[\bs z] - \bs z\|_{p,\bs w}$, is rather limited and decreases when
$p$ and $B$ increase with a maximum relative error of about $2.5\%$
between the true and estimated distortion at $B=3$ and $p=2$.

Inspired by relation \eqref{eq:eq_lemma_Qp_dist_bound}, we say that an
estimate $\bs x^*\in\Rbb^N$ of $\bs x$ sensed by the
model \eqref{eq:quant-sensing-model} satisfies the \emph{$p$-Distortion
Consistency} (or D$_p$C) if
\begin{align}
  \label{eq:Dp_consistency}
  \|\bs \Phi \bs x^* - \cl Q_p[\bs y]\|_{p,\bs w}\ \leq\
  \epsilon_p,\tag{\bf D$_p$C}
\end{align}
with the weights $w_i(p) = \cl G'\big(\cl Q_p[y_i])^{(p-2)/p}$.

The class of D$_p$C constraints
  has QC and DC as its limit cases.
\begin{lemma} 
\label{lemma:equiv-D2C-DinfC}
Given $\bs y = \cl Q[\bs \Phi \bs x]$, we have asymptotically in $B$
$$
{\rm D}_2{\rm C}\ \equiv\ {\rm DC}\qquad\text{and}\qquad {\rm D}_{\infty}{\rm C}\ \equiv\ {\rm QC}.
$$ 
\end{lemma}
\begin{proof}
Let $\bs x^* \in \Rbb^N$ be a vector to be tested with the DC, QC or
D$_p$C constraints. The first equivalence for $p=2$ is straightforward since $\bs w(2)=\bs 1$,
$\|\bs \Phi \bs x^* - \cl Q_p[\bs y]\|_{p,\bs w} =\|\bs \Phi \bs x^* -
\cl Q[\bs y]\|_{2}$ and $\epsilon^2_2=\epsilon^2_{\rm DC} =
\tfrac{2^{-2B}}{12}\,\fnorm{\pdf_0}_{1/3}$ from \eqref{eq:panter-dite}.  

For the second, we use the fact that $\bs y = \cl Q[\bs \Phi \bs x]$
is fixed by the sensing model \eqref{eq:quant-sensing-model}. Let us
denote by $k(i)$ the index of the bin to which $\cl Q_p[y_i]$ belongs
for $1\leq i \leq M$. Since $\|\bs\Phi \bs x\|_\infty$ is fixed, and
because relation \eqref{lemma:useful-ineq-eq2} in
Lemma~\ref{lemma:useful-ineq} implies that the amplitude of the first
or of the last $\Theta(B^{-1/2}2^{(1-\beta)B})$ thresholds grow faster
than $T=\Theta(\sqrt{\beta B})$ for $0<\beta<1$, there exists necessarily
a $B_0\geq 0$ such that $-T(B) \leq t_{k(i)} \leq t_{k(i)+1} \leq
T(B)$ for all $B\geq B_0$ and all $1\leq i\leq M$.

Writing $\bs W_{\!\!p} = \diag(\bs w(p))$, we can use the equivalence
$\|\cdot \|_\infty \leq \|\cdot \|_p \leq M^{1/p}\,\|\cdot \|_\infty$
and the squeeze theorem on the following limit:
$$
\lim_{p\to\infty}\|\bs \Phi \bs x^* - \cl Q_p[\bs y]\|_{p,\bs w(p)}
 =\ \lim_{p\to\infty}\|\bs W_{\!\!p}\big(\Phi \bs x^* - \cl
Q_p[\bs y]\big)\|_{p}\\
 =\ \lim_{p\to\infty}\,\|\bs W_{\!\!p}\big(\Phi \bs x^* -
\cl Q_p[\bs y]\big)\|_{\infty}.
$$

Moreover, since for $B\geq B_0$ and for all $1\leq i\leq M$ the bin $\cl R_{k(i)}$ is
finite, the limit
$$
\lim_{p\to\infty}\,\cl G'(\cl Q_p[y_i])^{(p-2)/p}\,\big|(\bs \Phi \bs
x^*)_i - \cl Q_p[y_i]\big|
$$ 
exists and is finite. Therefore, from the continuity of the
  $\max$ function applied on the $M$ components of vectors in $\Rbb^M$, we find
\begin{align*}
\lim_{p\to\infty}\|\bs \Phi \bs x^* - \cl Q_p[\bs y]\|_{p,\bs w(p)}&=\ \lim_{p\to\infty}\,\max_{i}\,\cl G'(\cl Q_p[y_i])^{(p-2)/p}\,\big|(\bs \Phi \bs x^*)_i - \cl Q_p[y_i]\big|\\
&=\ \max_{i}\,\lim_{p\to\infty}\ \cl G'(\cl Q_p[y_i])^{(p-2)/p}\,\big|(\bs \Phi \bs x^*)_i - \cl Q_p[y_i]\big|\\
&=\ \max_{i}\,\cl G'(\cl Q_\infty(y_i))\,\big|(\bs \Phi \bs x^*)_i - \cl Q_\infty(y_i)\big|.
\end{align*}

For $B\geq B_0$, \eqref{lemma:useful-ineq-eq8} provides $\cl G'(\cl
Q_\infty(y_i)) \simeq_B \tfrac{\binw}{\gbinw_{k(i)}}$, so that, if we
impose $\lim_{p\to\infty}\|\bs \Phi \bs x^* - \cl Q_p[\bs y]\|_{p,\bs
  w(p)} \leq \epsilon_{\rm QC} = \binw/2$, we get asymptotically in
$B$
$$
\max_{i} \tinv{\gbinw_{k(i)}}\,\big|(\bs
\Phi \bs x^*)_i - \cl Q_\infty(y_i)\big|\ \mathop{\lesssim}_B\ \tinv{2},
$$
which is equivalent to imposing $(\bs
\Phi \bs x^*)_i \in \cl R_{k(i)}$, \ie the
Quantization Constraint.  
\end{proof}

\section{Weighted $\ell_p$ Fidelities in Compressed Sensing and General Reconstruction Guarantees}
\label{sec:weight-fidel-guar}

The last section has provided us some weighted $\ell_{p,\bs w}$
constraints, with appropriate weights $\bs w$, that can be used for
stabilizing the reconstruction of a signal observed through the
quantized sensing model \eqref{eq:quant-sensing-model}. 
We now turn to studying the stability of $\ell_1$-based decoders integrating these weighted
$\ell_{p,\bs w}$-constraints as data fidelity. 
We will highlight also the requirements that the sensing matrix must fulfill to ensure this
stability. We then then apply this general stability result to additive heteroscedastic GGD noise, where weighing can be view as a variance stabilization transform. Section~\ref{sec:dequ-with-gener} will
later instantiate the outcome of this section to the particular case of QCS.

\subsection{Generalized Basis Pursuit DeNoise }
\label{sec:dequ-with-gener}

Given some positive weights $\bs w \in \Rbb^M$ and $p\geq 2$, we
study the following general minimization program, coined General Basis
Pursuit DeNoise (GBPDN),
\begin{equation*}
  \Delta_{p, \bs w}(\bs y,\bs\Phi,\epsilon)\ = \ \Argmin_{\bs u\,\in\,\Rbb^N} \|\bs u\|_1
  \ {\rm
  s.t.}\ \|\bs y - \bs\Phi \bs u\|_{p,\bs w}\leq
  \epsilon, \eqno{(\textrm{\bf\small GBPDN$(\ell_{p,\bs w})$})}
\end{equation*}
where $\|\!\cdot\!\|_{p,\bs w}$ is the weighted $\ell_p$-norm defined in the previous section.
Note that BPDN is special case of GBPDN corresponding to $p=2$ and $\bs
w=\bs 1$. The Basis Pursuit DeQuantizers (BPDQ) introduced in
\cite{Jacques2010} are associated to $p\geq 1$ and $\bs w=\bs 1$,
while the case $p=1$ and $\bs w=\bs 1$ has also been covered in
\cite{fuchs2009fast}.

We are going to see that the stability of GBPDN$(\ell_{p,\bs w})$ is
guaranteed if $\bs \Phi$ satisfies a particular instance of the
following general isometry property.

\begin{definition}
  Given two normed spaces $\mathcal{X}=(\Rbb^M,\norm{\cdot}_{\cl X})$
  and $\mathcal{Y}=(\Rbb^N,\norm{\cdot}_{\cl Y})$ (with $M< N$), a matrix
  $\bs\Phi\in\Rbb^{M\times N}$ satisfies the Restricted Isometry
  Property from $\mathcal{X}$ to $\mathcal{Y}$ at order
  $K\in\Nbb$, radius $0\leq\delta<1$ and for a normalization
  $\mu>0$, if for all $\bs
  x\in\Sigma_K$,
\begin{equation}
  \label{eq:rip-p}
  (1-\delta)^{1/\kappa}\,\norm{\bs x}_{\mathcal{Y}} \leq \tinv{\mu} \norm{\bs \Phi \bs x}_{\mathcal{X}} \leq
  (1+\delta)^{1/\kappa}\,\norm{\bs x}_{\mathcal{Y}},
\end{equation}
$\kappa$ being an exponent function of the geometries of $\cl X,\cl Y$. To lighten notation, we will write that
$\bs \Phi$ is RIP$_{\mathcal{X},\mathcal{Y}}(K,\delta,\mu)$.
\end{definition}

We may notice that the common RIP is equivalent to\footnote{Assuming the columns of $\bs \Phi$ are normalized to unit-norm.}
RIP$_{\ell_2^M,\ell_2^N}(K,\delta,1)$ with $\kappa=1$, while the RIP$_{p,q}$ introduced
earlier in \cite{Jacques2010} is equivalent to
RIP$_{\ell_p^M,\ell_q^N}(K,\delta,\mu)$ with $\kappa=q$ and $\mu$ depending only on $M$, $p$ and
$q$. Moreover, the RIP$_{p,K,\delta'}$ defined in \cite{Berinde2008}
is equivalent to the RIP$_{\ell_p^M,\ell_p^N}(K,\delta,\mu)$ with
$\kappa=1$, $\delta'=2\delta/(1-\delta)$ and $\mu=1/(1-\delta)$. Finally, the
Restricted $p$-Isometry Property proposed in
\cite{chartrand2008restricted} is also equivalent to the
RIP$_{\ell_p^M,\ell_2^N}(K,\delta,1)$ with $\kappa=p$.

In order to study the behavior of the GBPDN program, we are interested
in the embedding induced by $\bs\Phi$ in \eqref{eq:rip-p} of $\cl Y =
\ell_2^N$ into the normed space $\cl X = \ell_{p,\bs
  w}^M=(\Rbb^M,\norm{\!\cdot\!}_{p,\bs w})$, \ie we consider the
RIP$_{\ell_{p,\bs w}^M,\,\ell_2^N}$ property that we write in the
following as RIP$_{p,\bs w}$. The following theorem
  establishes that GBPDN provides stable recovery from distorted
  measurements, if the RIP$_{p,\bs w}$ holds. 
\begin{theorem}
  \label{prop:l2-l1-instance-optimality-GBPDN}
  Let $K\geq 0$, $2\leq p <\infty$ and $\bs \Phi\in\Rbb^{M\times N}$ be a
  RIP$_{p,\bs w}(s,\delta_s,\mu)$ matrix for
  $s\in\{K,2K,3K\}$
such that 
  \begin{equation}
    \label{eq:cond-on-delta-p}
    \delta_{2K}+\sqrt{(1+\delta_K)(\delta_{2K}+\delta_{3K})(p - 1)} < 1/3 ~.   
  \end{equation}
Then, for any signal $\bs x\in\Rbb^N$ observed according to the noisy
sensing model $\bs y=\bs \Phi \bs x + \bs \varepsilon$ with $\|\bs \varepsilon\|_{p,\bs w}\leq
  \epsilon$, the unique solution $\bs x^*=\Delta_{p, \bs w}(\bs y,\bs \Phi,
  \epsilon)$  obeys
\begin{equation}
\label{eq:GBPDN-l2-l1-inst_opt}
\|\bs x^* - \bs x\| \ \leq\ 4\,e_0(K)\ +\ 8\,\epsilon/\mu,
\end{equation}
where $e_0(K) = K^{-\frac{1}{2}}\,\|\bs x -
  \bs x_K\|_1$ is the $K$-term
  $\ell_1$-approximation error.
\end{theorem}
\begin{proof}
If $\bs \Phi$ is
RIP$_{p,\bs w}(s,\delta_s,\mu)$ for $s\in\{K,2K,3K\}$, then, by definition
of the weighted $\ell_{p,\bs w}$-norm, $\diag(\bs
w)\bs\Phi$ is RIP$_{\ell_{p}^M,\ell_2^N}(s,\delta_s,\mu)$. Since $\Delta_{p,\bs w}(\bs
y,\bs\Phi, \epsilon) \ = \Delta_{p}(\diag(\bs
w)\bs y,\diag(\bs
w)\bs\Phi, \epsilon)$, the
stability results proved in \cite[Theorem 2]{Jacques2010} for 
GBPDN$(\ell_{p})$ \footnote{Dubbed BPDQ in
  \cite{Jacques2010}.} shows that
$$
\|\bs x - \bs x^*\| \leq A_p\, e_0(K) + B_p\, \tfrac{\epsilon}{\mu},
$$
with $
A_p =
\tfrac{2(1+C_p-\delta_{2K})}{1-\delta_{2K}-C_p}$, $B_p =
\tfrac{4\sqrt{1+\delta_{2K}}}{1-\delta_{2K}-C_p}$ and $C_p \leq \sqrt{(1+\delta_K)(\delta_{2K}+\delta_{3K})(p - 1)}$ \cite{Jacques2010}.
It is easy to see that if \eqref{eq:cond-on-delta-p} holds, then $A_p
\leq 4$ and $B_p \leq 8$. 
\end{proof}

As we shall see shortly, this theorem may be used to characterize
the impact of measurement corruption due to both 
additive heteroscedastic GGD noise (Section~\ref{sec:case-heter-nois}) as well as those
induced by a non-uniform scalar quantization (Section~\ref{sec:dequ-with-gener}).
Before detailing these two sensing scenarios, we first address the
question of designing matrices satisfying the RIP$_{p,\bs w}$ for $2\leq
p < \infty$.

\subsection{Weighted Isometric Mappings}
\label{sec:weight-isom-mapp}

We will describe a random matrix construction that will satisfy the 
RIP$_{p,\bs w}$ for $1\leq p < \infty$. To quantify when this is possible,
we introduce some properties on the positive weights $\bs w$. 
\begin{definition}
  \label{def:bound-weight-moment}
  A weight generator $\cl W$ is a process (random or deterministic)
  that associates to $M\in\Nbb$ a weight vector $\bs w = \cl W(M)\in
  \Rbb^M$. This process is said to be of Converging Moments (CM) if
  for $p\geq 1$ and all $M\geq M_0$ for a certain $M_0>0$,
\begin{equation}
  \label{eq:bound-weight-moment}
  \rho_p^{\min} \ \leq\ M^{-1/p}\,\|\cl W(M)\|_p\ \leq\
  \rho_p^{\max},
\end{equation}
where $\rho_p^{\min}>0$ and $\rho_p^{\max}>0$ are, respectively, the
largest and the smallest values such that
\eqref{eq:bound-weight-moment} holds. In other words, a CM generator
$\cl W$ is such that $\|\cl W(M)\|^p_p=\Theta(M)$. By extension, we
say that the weighting vector $\bs w$ has the CM property, if it is generated by
some CM weight generator $\cl W$.
\end{definition}

The CM property can be ensured if $\lim_{M \to \infty} M^{-1/p}\,\|\bs
w\|_p$ exists, bounded and nonzero. It is also ensured if the weights
$\{w_i\}_{1\leq i\leq M}$ are taken (with repetition) from a finite
set of positive values. More generally,
if $\{w_i: 1\leq i\leq M\}$ are $\iid$ random variables, 
we have $M^{-1}\,\norm{\bs w}^p_p = \E|w_1|^p$ almost surely by the SLLN. Notice
finally that $\rho^{\max}_p\leq \|\bs
w\|_{\infty}=\rho^{\max}_{\infty}$ since $\|\bs w\|^p_p \leq M \|\bs
w\|_{\infty}^p$, and $\rho^{\min}_p\geq \min_i |w_i|$.

For a weighting vector $\bs w$ having the CM property, we define also
its \emph{weighting dynamic} at moment $p$ as the ratio
$$
\theta_p =
\big(\tfrac{\rho^{\max}_{\infty}}{\rho_p^{\min}}\big)^2.
$$
We will see later that $\theta_p$
directly influences the number of measurements required to guarantee
the existence of RIP$_{p,\bs w}$ random Gaussian matrices.

Given a weight vector $\bs w$, the following lemma (proved in
Appendix~\ref{sec:proof-lemma-strict-bound}) characterizes the
expectation of the $\ell_{p,\bs w}$-norm of a random Gaussian vector.
\begin{lemma}[\bf Gaussian $\ell_{p,w}$-Norm Expectation]
\label{lem:strict-bounds-mu_p}
If $\bs \xi\sim \cl N^M(0,1)$ and if the weights $\bs w$ have the CM property, 
then, for $1\leq p<\infty$ and $\cl Z \sim \cl N(0,1)$,
\begin{equation*}
  \big(1\ +\ 2^{p+1}\,\theta_p^p\, M^{-1}\big)^{\inv{p}-1}\,(\E\norm{\bs
  \xi}_{p,\bs w}^p)^{\inv{p}}\ 
  \leq\ \E\norm{\bs \xi}_{p,\bs w}\ \leq\ (\E\norm{\bs \xi}_{p,\bs
  w}^p)^{\inv{p}} = (\E|\cl Z|^p)^{1/p}\|\bs w\|_p.
\end{equation*}
In particular, $\E\|\bs \xi\|_{p,\bs w} \simeq_M \nu_p\,\|\bs w\|_p
\geq \nu_p\, M^{1/p}\,\rho_p^{\min}$, with $\nu_p^p \eqdef
\E|\cl Z|^p=2^{p/2}\pi^{-1/2}\Gamma(\tfrac{p+1}{2})$.
\end{lemma}

With an appropriate modification of \cite[Proposition~1]{Jacques2010},
we can now prove the existence of random Gaussian RIP$_{p,\bs w}$
matrices (see Appendix \ref{sec:proof-prop-grip-gauss}).
\begin{proposition}[\bf RIP$_{p,\bs w}$ Matrix Existence]
\label{prop:grip-gauss}
Let $\bs \Phi\sim \cl N^{M\times N}(0,1)$ and some CM weights $\bs
w\in\Rbb^M$. Given $p\geq 1$ and $0\leq \eta < 1$, then there exists a
constant $c>0$ such that $\bs\Phi$ is RIP$_{p,\bs w}(K,\delta,\mu)$ with probability higher than
$1-\eta$ when we have jointly $M\geq 2\,(2\theta_p)^p$, and
\begin{eqnarray}
\label{eq:SGR-RIP-measur-bound}
M^{2/\!\max(2,\,p)}&\geq&c\,\delta^{-2}\,\theta_p\,
\big(K\log [e\tfrac{N}{K}(1+12\delta^{-1})]\ +\ \log\tfrac{2}{\eta}\big).
\end{eqnarray}
Moreover, the value
$\mu=\mu(\lpn{p}{\bs w}{M},\ell^N_2)$ in \eqref{eq:rip-p} is given by
$\mu=\E\norm{\bs \xi}_{p,\bs w}$ for a random vector $\bs\xi\sim \cl
N^M(0,1)$.
\end{proposition}

The RIP normalizing constant $\mu$ can be bounded owing to Lemma \ref{lem:strict-bounds-mu_p}. 

\begin{remark}
  In the light of Proposition~\ref{prop:grip-gauss}, assumption
  \eqref{eq:cond-on-delta-p} becomes reasonable since following the
  simple argument presented in \cite[Appendix~B]{Jacques2010} the
  saturation of requirement \eqref{eq:SGR-RIP-measur-bound} implies that
  $\delta_{K}$ decays as $O(\sqrt{K\log M}/M^{1/p})$ for RIP$_{p,\bs
    w}$ Gaussian matrices.  Therefore, for any value $p$, it is always
  possible to find a $M$ such that \eqref{eq:cond-on-delta-p}
  holds. However, this is only possible for high oversampling
  situation, \ie for $\Omega(
  (K\log N/K)^{p/2})$ measurements.
\end{remark}

\subsection{GBPDN stabilizes Heteroscedastic GGD Noise}
\label{sec:case-heter-nois}

Consider the following general signal sensing model
\begin{equation}
\label{eq:signal-cs-sensing}
\bs y\ =\ \bs \Phi \bs x\ +\ \bs \varepsilon ~,
\end{equation}
where $\bs \varepsilon\in \Rbb^M$ is the noise vector. For
heteroscedastic GGD noise, each 
$\varepsilon_i$ follows a zero-mean ${\rm GGD}(0,\alpha_i,p)$
distribution with pdf $\propto \exp(-|t/\alpha_i|^p)$, where $p > 0$
is the shape parameter (the same for all $\varepsilon_i$'s), and
$\alpha_i > 0$ the scale parameter \cite{varanasi1989pgg}. It is
obvious that
$$
\E\bs \varepsilon=\bs 0\quad \text{and}\quad \E(\bs \varepsilon\bs \varepsilon^\Tr) = \Gamma(3/p)
(\Gamma(1/p))^{-1}\diag(\alpha_1^2,\,\cdots,\alpha_M^2).
$$

If one sets the weights to $w_i=1/\alpha_i$ in GBPDN$(\ell_{p,\bs
  w})$, it can be seen that the associated constraint corresponds
precisely to the negative log-likelihood of the joint pdf of $\bs
\varepsilon$.
As detailed below, introducing these non-uniform weights $w_i$ leads
to a reduction in the error of the reconstructed signal, relative to
using constant weights. Without loss of generality, we here restrict our analysis to strictly
$K$-sparse $\bs x \in \Sigma_K$, and assume knowledge of bounds
(estimators) for the $\ell_p$ and the $\ell_{p,\bs w}$ norms used for
characterizing $\bs \varepsilon$, \ie we know that $\|\bs
\varepsilon\|_p\simeq_M \epsilon$ and $\|\bs \varepsilon\|_{p,\bs
  w}\simeq_M \epsilon_{\rm st}$ for some $\epsilon,\epsilon_{\rm
  st}>0$ to be detailed later.

In this case, if the random matrix $\bs \Phi\sim\cl N^{M\times N}(0,1)$ is
RIP$_{p,\bs w}(K,\delta,\mu)$ for $p \geq 2$, with $\mu=\E\|\bs \xi\|_p$
for $\bs \xi\sim\cl N^{M}(0,1)$, Theorem~\ref{prop:l2-l1-instance-optimality-GBPDN} asserts that
$$
\|\bs x^* - \bs x\|\ \leq\ B_p\,\epsilon/\mu,
$$
for $\bs x^*=\Delta_{p, \bs 1}(\bs y,\bs \Phi,
  \epsilon)$ and $B_p \simeq_M 8$.
Conversely, for the weights to $w_i=1/\alpha_i$, and assuming $\bs \Phi$ being
RIP$_{p,\bs w}(K,\delta',\mu_{\rm st})$ with
$\mu_{\rm st}=\E\|\bs \xi\|_{p,\bs w}$, we get
$$
\|\bs x^*_{\rm st} - \bs x\|\ \leq\ B'_p\,\epsilon_{\rm st}/\mu_{\rm st},
$$
for $\bs x^*_{\rm st}=\Delta_{p, \bs w}(\bs y,\bs \Phi, \epsilon)$ and
$B'_p \simeq_M 8$.

When the number of measurements $M$ is large, using classical GGD absolute moments formula, the two bounds $\epsilon$
and $\epsilon_{\rm st}$ can be set close to 
$\epsilon^p \mathop{\simeq}_M \sum_i \E|\varepsilon_i|^p = \|\bs \alpha\|_p^p/p$ and $\epsilon_{\rm st}^p \mathop{\simeq}_M \sum_i w_i^p\,\E|\varepsilon_i|^p = M/p$. Moreover, using Lemma~\ref{lem:strict-bounds-mu_p},
$\mu^p \mathop{\simeq}_M \sum_i \E|\xi_i|^p = M \E|\cl Z|^p$ and $\mu_{\rm
  st}^p \mathop{\simeq}_M \E|\cl Z|^p\,\|\bs w\|^p_p$, where $\cl Z \sim \cl N(0,1)$.

\begin{proposition}
  \label{prop:whitening-heter-ggd-nois}
For an additive heteroscedastic noise $\bs \varepsilon\in\Rbb^M$ such that 
$\varepsilon_i \sim_{\rm iid} {\rm GGD}(0,\alpha_i,p)$, setting $w_i = 1/\alpha_i$ provides
$\epsilon_{\rm st}^p/\mu_{\rm st}^p\ \lesssim_M\
\epsilon^p/\mu^p$. Therefore, asymptotically in $M$, GBPDN$(\ell_{p,\bs w})$ has a smaller
reconstruction error compared to GBPDN$(\ell_{p})$ when estimating
$\bs x$ from the sensing model \eqref{eq:signal-cs-sensing}. 
\end{proposition}
\begin{proof}
  Let us observe that $ \epsilon_{\rm st}^p/\mu_{\rm st}^p \simeq_M
  M (p\,\E|\cl Z|^p\,\|\bs w\|^p_p)^{-1} = (p\,\E|\cl Z|^p)^{-1}
  (\tinv{M}\sum_i \tinv{\alpha_i^{p}})^{-1}$. By the Jensen inequality, 
  $(\tinv{M}\sum_i \tinv{\alpha_i^{p}})^{-1} \leq \tinv{M}\sum_i\alpha_i^p$,
  so that $ \epsilon_{\rm st}^p/\mu_{\rm st}^p
  \lesssim_M \tinv{p}\,(\E|\cl Z|^p)^{-1} \|\bs \alpha\|^{p}_p/M =
  \epsilon^p/\mu^p$.
\end{proof}

The price to pay for this stabilization is an increase of the weighting
dynamic $\theta_p=(\tfrac{\rho^{\max}_\infty}{\rho^{\min}_p})^2$
defined in Proposition~\ref{prop:grip-gauss}, which implies an
increase in the number of measurements $M$ needed to ensure that
the RIP$_{p,\bs w}(K,\delta,\mu)$ is satisfied.

\begin{example}
  Let us consider a simple situation where the $\alpha_i$'s take only
  two values, \ie $\alpha_i \in \{1, H\}$ for some $H \geq 1$. Let us
  assume also that the
  proportion of $\alpha_i$'s equal to $H$ converges to 
  $r\in [0,1]$ with $M$ as $|\inv{M}\#\{i: \alpha_i=
  H\}\ -\ r| = O(M^{-1})$. In this case, the stabilizing weights are $w_i = 1/\alpha_i \in
  \{1,1/H\}$. An easy computation provides
\begin{align*}
\fs E&\eqdef \tfrac{\epsilon^p}{\mu^p}\ \mathop \simeq_{M}\
\tinv{p}\nu_p^{-p}\, \big(\,r\,H^p
+ (1-r)\,\big),\\
\fs E_{\rm st}&\eqdef \tfrac{\epsilon_{\rm st}^p}{\mu_{\rm st}^p}\ \mathop \simeq_{M}\
\tinv{p}\nu_p^{-p}\, \big(\,r\,H^{-p}
+ (1-r)\,\big)^{-1},  
\end{align*}
so that, the ``stabilization gain'' with respect to an unstabilized setting can be
quantified by the ratio 
$$
(\tfrac{\fs E}{\fs E_{\rm st}})^{\inv{p}}\ {\mathop \simeq_{M}}\ \big(\,r\,H^{-p}
+ (1-r)\,\big)^{\inv{p}}\big(\,r\,H^{p}
+ (1-r)\,\big)^{\inv{p}}\ {\mathop \simeq_{M,H}}\ \big(r\,(1-r)\big)^{\inv{p}}\,H.
$$ 
We see that the stabilization provides a clear gain which increases as the measurements get very unevenly corrupted, \ie when $H$ is large. Interestingly, the higher $p$ is,
the less sensitive is this gain to $r$. We also observe that the overhead in the
number of measurements between the stabilized and the unstabilized
situations is related to 
$$
\theta^{p/2}_p\ = \big(\tfrac{\rho^{\max}_\infty}{\rho^{\min}_p}\big)^{p}\ {\mathop \simeq_{M}}\ \big(\,r\,H^{-p}
+ (1-r)\,\big)^{-1}\ {\mathop \simeq_{M,H}}\ (1-r)^{-1}.
$$

The limit case where $H\gg 1$ can be interpreted as ignoring $r$ percent of the
measurements in the data fidelity constraint, keeping only those for which the noise is not
dominating. In that case, the sufficient condition \eqref{eq:SGR-RIP-measur-bound} in Proposition~\ref{prop:grip-gauss} for $\bs\Phi$ to be RIP$_{p,\bs w}$ tends to $\theta_p^{-p/2}M = (1-r)M = \Omega\big((K \log N/K)^{p/2}\big)$ which
  is consistent with the fact that on average only fraction $1-r$ of the
  $M$ measurements significantly participate to the CS scheme, \ie
  $M'=(1-r)M$ must satisfy the common RIP requirement.
  For $p=2$, this is somehow related to the democratic property of RIP
  matrices \cite{LasBouDav::2009::Demcracy-in-action}, \ie the fact that a reasonable number of rows can be discarded 
  from a matrix while preserving the RIP. This property was successfully used for
  discarding saturated CS measurements in the case of a limited
  dynamic quantizer \cite{LasBouDav::2009::Demcracy-in-action}. 
\end{example}

\section{Dequantizing with Generalized Basis Pursuit DeNoise }
\label{sec:dequ-with-gener}

Let us now instantiate the use of GBPDN to the reconstruction of signals in the QCS scenario
  defined in Section\ref{sec:non-uniform-scalar}. Under the quantization formalism defined in
  Lemma~\ref{lemma:bound-lpw-gaussian-vector} and for Gaussian
  matrices $\bs \Phi$, the factor $\epsilon/\mu$ in \eqref{eq:GBPDN-l2-l1-inst_opt} can be shown to
  decrease as $1/\sqrt{p+1}$ asymptotically in $M$ and $B$.  This
  asymptotic and \emph{almost sure} result which relies on the SLLN (see Appendix \ref{proof:towards-quant-cons}) suggests
  increasing $p$ to the highest value allowed by
  \eqref{eq:cond-on-delta-p} in order to decrease the GBPDN
  reconstruction error.

\begin{proposition}[\bf Dequantizing Reconstruction Error]
  \label{prop:towards-quant-cons}
  Given $\bs x \in \Rbb^N$ and $\bs \Phi \sim \cl N^{M\times N}(0,1)$, 
  assume that the entries of $\bs z = \bs\Phi\bs x$ are $\iid$ realizations from $\cl Z \sim \cl
  N(0,\sigma_0^2)$. We take the corresponding optimal compressor
  function $\cl G$ defined in \eqref{eq:optimal-compressor} and the
  $p$-optimal $B$-bits scalar quantizer $\cl Q_p$ as defined in
  \eqref{eq:Qp}. Then, the ratio ${\epsilon}/{\mu}$ given in
  \eqref{eq:GBPDN-l2-l1-inst_opt} is asymptotically and almost surely bounded by
$$
\tfrac{\epsilon}{\mu}\ \lesssim_{B,M}\ c'\,2^{-B}\,
\,\tfrac{(p+1)^{-\inv{2p}}}{\sqrt{p+1}}\leq c'\,\frac{2^{-B}}{\sqrt{p+1}}.
$$
with $c'=(9/8)(e\pi/3)^{1/2}$.
\end{proposition}

Notice that, under HRA and for large $M$, it is possible to provide a
rough estimation of the weighting dynamic $\theta_p$
when the weights are those provided by the D$_p$C constraints. Indeed,
since $w_i(p)=\cl G'(\cl Q_p[y_i])^{(p-2)/p}$ and $\cl G' =
\gamma_{0,\sqrt 3\sigma_0}$, we find 
\begin{align*}
\|\bs w\|^p_p&= \sum_i \cl G'(\cl Q_p[y_i])^{p-2}\ \simeq_M\ M\,\sum_k \cl
G'^{p-2}(\omega_{k,p})\,p_k\\
&\simeq_{B,M}\, M\, (2\pi 3\sigma_0^2)^{(2-p)/2}(2\pi \sigma_0^2)^{-1/2} \sum_k
\gbinw_k\,\exp(-\tinv{2}\omega_{k,p}^2\tfrac{p+1}{3\sigma^2_0})\\
&\simeq_{B,M}\, M\, (2\pi 3\sigma_0^2)^{(2-p)/2}(2\pi \sigma_0^2)^{-1/2}
(2\pi \tfrac{3\sigma_0^2}{p+1})^{1/2}\\
&=\ M\,(2\pi\sigma_0^2)^{(2-p)/2} 3^{(3-p)/2} (p+1)^{-1/2},
\end{align*}
where we recall that $p_k=\int_{\cl R_k}\pdf_0(t) \ud t \simeq_{B}
\pdf_0(c')\gbinw_k$, for any $c' \in \cl R_k$ (see the proof of Lemma~\ref{lemma:bounding-a-tail-after-T}).

Moreover, using \eqref{lemma:useful-ineq-eq8} and since one of the two smallest quantization bins is $\cl R_{\cl
  B/2}=[0,\tau_{\cl B/2})$,
$$
\|\bs w\|^p_\infty \simeq_B (\binw/\gbinw_{\cl B/2})^{p-2} = (\binw/\cl G^{-1}(1/2 +
\binw))^{p-2} \simeq_B (2\pi3\sigma_0^2)^{(2-p)/2}.
$$
Therefore, estimating $\theta_p^{p}$ with $M^2\|\bs w\|^{2p}_\infty/\|\bs
w\|^{2p}_p$, we find
$$
\theta_p^{p/2}\ \simeq_{B,M}\ \sqrt{(p+1)/3}.
$$

Therefore, at a given $p\geq 2$, since (\ref{eq:SGR-RIP-measur-bound})
involves that $M$ evolves like $\Omega(\theta_p^{p/2} (K \log
N/K)^{p/2})$, using the weighting induced by GBPDN($\ell_{p,\bs w}$)
requires collecting $\sqrt{(p+1)/3}$ times more measurements than
GBPDN($\ell_{p}$) in order to ensure the appropriate RIP$_{p, \bs w}$ property.
This represents part of the price to pay for guaranteeing bounded
reconstruction error by adapting to non-uniform quantization.
\medskip

\noindent \textbf{Dequantizing is Stabilizing Quantization Distortion:}

In connection with the procedure developed in
Section~\ref{sec:case-heter-nois}, the weights and the $p$-optimal levels
introduced in Lemma~\ref{lemma:bound-lpw-gaussian-vector} can be
interpreted as a ``stabilization'' of the quantization distortion seen as a heteroscedastic noise.
This means that, asymptotically in $M$, selecting these weights and levels, \emph{all quantization regions $\cl R_k$ contribute equally to the $\ell_{p,{\bs w}}$ distortion measure}.

To understand this fact, we start by studying the following relation shown in the proof of Lemma~\ref{lemma:bound-lpw-gaussian-vector}
(see Appendix~\ref{proof-lemma-bound-lpw-gaussian-vector}):
\begin{align}
\label{eq:contrib-to-lpw-distor}
\|\cl Q_p[\bs z] - \bs z\|^p_{p,\bs w}\ &\mathop{\simeq}_{M}\ M\,
\sum_{k}\ [\cl G'(\omega_{k,p})]^{p-2} \int_{\cl R_k}|t -
\omega_{k,p}|^p\,\pdf_0(t)\,\ud t.
\end{align}
Using the threshold $T(B)=\Theta(\sqrt B)$ and $\cl T=[-T(B),T(B)]$ as defined in Lemma
\ref{lemma:useful-ineq}, the proof of
Lemma~\ref{lemma:bounding-a-tail-after-T} in
Appendix~\ref{proof-lemma-bound-lpw-gaussian-vector} shows that
\begin{align}
\label{eq:contrib-to-lpw-distor}
\|\cl Q_p[\bs z] - \bs z\|^p_{p,\bs w}\ &\mathop{\simeq}_{M,B}\ M\,
\sum_{k: \cl R_k \subset \cl T}\ [\cl G'(\omega_{k,p})]^{p-2} \int_{\cl R_k}|t -
\omega_{k,p}|^p\,\pdf_0(t)\,\ud t,\\
&\mathop{\simeq}_{B}\ M\,
\sum_{k: \cl R_k \subset \cl T}\ [\cl G'(\omega_{k,p})]^{p-2} \tfrac{\gbinw_k^{p+1}}{(p+1)\,2^p} 
\,\pdf_0(\omega_{k,p}),
\end{align}
using \eqref{lemma:useful-ineq-eq7}. However, using
\eqref{lemma:useful-ineq-eq8} and the relation $\cl G' =
\pdf_0^{1/3}/\fnorm{\pdf_0}^{1/3}_{1/3}$, we find
$\gbinw^3_k\,\pdf_0(\omega_{k,p}) \simeq_{B} \binw^3\, \fnorm{\pdf_0}_{1/3}$.
Therefore, each term of the sum in~\eqref{eq:contrib-to-lpw-distor}
provides a contribution
$$
[\cl G'(\omega_{k,p})]^{p-2} \tfrac{\gbinw_k^{p+1}}{(p+1)\,2^p} 
\,\pdf_0(\omega_{k,p})\ \mathop{\simeq}_{B,M}\ \fnorm{\pdf_0}_{1/3}\,\tfrac{\binw^{p+1}}{(p+1)2^p},
$$    
which is independent of $k$!
This phenomenon is well known for $p=2$ and may actually serve for
defining $\cl G'$ itself \cite{gray1998quantization}. The fact that
this effect is preserved for $p\geq 2$ is a surprise for us.

\section{Numerical Experiments}
\label{sec:numer-exper}

We first describe how to numerically solve the GBPDN optimization
problem using a primal-dual convex optimization scheme, then
illustrate the use of GBPDN for stabilizing heteroscedastic Gaussian
noise on the CS measurements. Finally, we apply GBPDN for
reconstructing signals in the quantized CS scenario described in
Section~\ref{sec:non-uniform-scalar}.

\subsection{Solving GBPDN}
\label{sec:whitening-example}
The optimization problem GBPDN$(\ell_{p,\bs w})$ is a special instance of the general form
\begin{equation}
\label{eq:minfgL}
\min_{\bs u \in \Rbb^N} f(\bs u) + g({\bs L \bs u}) ~,
\end{equation}
where $f$ and $g$ are closed convex functions that are not infinite
everywhere (\ie proper functions), and $\bs L = \diag(\bs w) \bs \Phi$
is a bounded linear operator, with $f(\bs u)\eqdef\norm{\bs u}_1$, and
$g(\bs v) \eqdef \imath_{{\mathbb B}_p^\epsilon}(\bs v - \bs y)$ where
$\imath_{{\mathbb B}_p^\epsilon}(\bs v)$ is the indicator function of
the $\ell_p$-ball ${\mathbb B}_p^\epsilon$ centered at zero and of
radius $\epsilon$, \ie $\imath_{{\mathbb B}_p^\epsilon}(\bs v)=0$ if
$\bs v \in {\mathbb B}_p^\epsilon$ and $+ \infty$ otherwise. For the
case of GBPDN$(\ell_{p,\bs w})$, both $f$ and $g$
are non-smooth but the associated proximity operators (to be defined shortly) can be computed
easily. This will allow to minimize the GBPDN$(\ell_{p,\bs w})$
objective by calling on proximal splitting algorithms.

Before delving into the details of the minimization splitting
algorithm, we recall some results from convex analysis. The
\emph{proximity operator} \cite{moreau1962fcd} of a proper closed
convex $f$ is defined as the unique solution
\[
\prox_{f}(\bs u) = \argmin_{\bs z} \tinv{2}\norm{\bs z - \bs u}^2 + f(\bs z).
\]
If $f=\imath_C$ for some closed convex set $C$, $\prox_{f}$ is
equivalent to the orthogonal projector onto $C$, $\proj_C$. $f^*$ is
the \textit{Legendre-Fenchel conjugate} of $f$. For $\lambda > 0$, the
proximity operator of $\lambda f^*$ can be deduced from that of
$f/\lambda$ through Moreau's identity
\[
\prox_{\lambda f^*}(\bs u) = \bs u  - \lambda \prox_{\lambda^{-1} f}(\bs u/\lambda) ~.
\]

Solving \eqref{eq:minfgL} with an arbitrary bounded linear operator
$\bs L$ can be achieved using primal-dual methods motivated by the
classical Kuhn-Tucker theory. Starting from methods to solve saddle
function problems such as the Arrow-Hurwicz method
\cite{ArrowHurwicz58}, this problem has received a lot of attention recently,
\eg \cite{ChenTeboulle94,Chambolle2010,Briceno-AriasCombettes11}.
In this paper, we use the relaxed Arrow-Hurwicz algorithm as revitalized recently in
\cite{Chambolle2010}. Adapted to our problem, its steps are summarized in Algorithm~\ref{algo:gbpdn}.
\begin{algorithm}[h]
  \normalsize
  \noindent{\bf{Inputs:}} Measurements $\bs y$, sensing matrix $\bs \Phi$, weights $\bs w$.\\
  \noindent{\bf{Parameters:}} Iteration number $N_{\mathrm{iter}}$, $\theta \in [0,1]$, step-sizes $\sigma > 0$ and $\tau > 0$ with $\tau\sigma\norm{\bs w}_\infty^2\norm{\bs \Phi}^2 < 1$. \\
  \noindent{\bf{Main iteration:}} \\
  \For{$k=0$ {\bf{to}} $N_{\mathrm{iter}}-1$}{
  \begin{itemize}
  \item Update the dual variable:
  	\[
	{\bs v}_{k+1} = \prox_{\sigma g^*}({\bs v}_{k} + \sigma {\bs L} \overline{\bs u}_{k}) ~.
	\]
   \item Update the primal variable:
  	\[
	{\bs u}_{k+1} = \prox_{\tau f}({\bs u}_{k} - \tau {\bs L}^\Tr {\bs v}_{k+1}) ~.
	\]
   \item Approximate extragradient step:
    \[
	\overline{\bs u}_{k+1} =  {\bs u}_{k+1} + \theta({\bs u}_{k+1} - {\bs u}_{k})~.
	\]
  \end{itemize}
}
  \noindent{\bf{Output:}} Signal $\bs u_{N_{\mathrm{iter}}}$.
  \caption{Primal-dual scheme for solving GBPDN$(\ell_{p,\bs w})$.}
  \label{algo:gbpdn}
\end{algorithm}

A sufficient condition for the sequences of Algorithm~\ref{algo:gbpdn}
to converge is to choose $\sigma$ and $\tau$ such that
$\tau\sigma\norm{\bs w}_\infty^2\norm{\bs \Phi}^2 < 1$. It has been
shown in \cite[Theorem 1]{Chambolle2010} that under this condition and
for $\theta=1$, the primal sequence $(\bs u_k)_{k\in\mathbb{N}}$
converges to a (possibly strict) global minimizer of
GBPDN$(\ell_{p,\bs w})$, with the rate $O(1/k)$ in ergodic sense on the partial
duality gap.

\paragraph*{Proximity operator of $f$} For $f(\bs u)=\|\bs u\|_1$, $\prox_{\tau f}(\bs u)$ is the popular
component-wise soft-thresholding of $\bs u$ with threshold $\tau$.

\paragraph*{Proximity operator of $g$} Recall that $g(\bs v) = \imath_{{\mathbb B}_p^\epsilon}(\bs v - \bs y)$. Using Moreau's identity above, and proximal calculus rules for translation and scaling, we have 
\[
\prox_{\sigma g^*}(\bs v) = {\bs v} - \sigma {\bs y} - \proj_{\imath_{{\mathbb B}_p^{\sigma\epsilon}}}(\bs v - \sigma \bs y) ~.
\]
It remains to compute the orthogonal projection $\proj_{{\mathbb B}^1_p}$ to get $\proj_{{\mathbb B}^{\sigma\epsilon}_p} = \sigma\epsilon\proj_{{\mathbb B}^1_p}(\cdot/(\sigma\epsilon))$. For $p=2$ and $p=+\infty$, this projector has an easy closed form. For $2<p<+\infty$, we used the Newton method we proposed in \cite{Jacques2010} for solving the related Karush-Kuhn-Tucker system which is reminiscent of the strategy underlying sequential quadratic programming.

\subsection{Gaussian Noise Stabilization Illustration}
\label{sec:whitening-example}

We explore numerically the impact of using non-uniform weights (\eg
stabilizing the measurement noise) for signal reconstruction when the
CS measurements are corrupted by heteroscedastic Gaussian noise, as
discussed in Section~\ref{sec:case-heter-nois}. This illustrates for
$p=2$ both the gain induced by stabilizing the sensing noise and the
increase of measurements necessary for observing this gain.

In this illustration, we set the problem dimensions to $N=1024$,
$K=16$, and let the oversampling factor be in
$M/K\in\{5,10,\cdots,50\}$. The $K$-sparse unit norm signals were
generated independently according to a Bernoulli-Gaussian mixture model with
$K$-length support picked uniformly at random in $[N]$, and the
non-zero signal entries drawn from $\cl N(0,\sigma_s^2)$ with
$\sigma_s^2\simeq 1/K$.  Noisy measurements were simulated by setting
$\bs y = \bs\Phi\bs x+ \bs \varepsilon$, with
$\varepsilon_i\sim_{\iid} \cl N(0,\sigma_i^2)$ and $\bs\Phi\sim \cl
N^{M\times N}(0,1)$. The heteroscedastic behavior of $\bs \varepsilon$
has been designed so that $\sigma_i \sim_{\iid} \cl
U([\sigma_0-\delta_0,\sigma_0+\delta_0])$ with $\sigma_0=0.1$ and
$\delta_0=0.6\,\sigma_0$.

Two reconstruction methods were tested: one \emph{with} and the
other \emph{without} stabilizing the noise variance. In the first
case, the weights have been set to $w_i=1/\sigma_i$, while in the
second $\bs w=\bs 1$. Since the purpose of this analysis is not
focused on the design of efficient noise power estimators, $\epsilon$
and $\epsilon_{\rm st}$ have been simply set by an oracle to
$\epsilon_{\rm st} = \|\bs y - \bs\Phi\bs x\|_{2,\bs w}$ and
$\epsilon=\|\bs y - \bs\Phi\bs x\|_2$.

Given the parameters above, we compute the weighting dynamic $\theta_p
\simeq_M \frac{M\,\E \|\bs w\|^2_\infty}{\E \|\bs w\|^2_2} =
\frac{\sigma_0 + \delta_0}{\sigma_0 - \delta_0} = 4$, and the average
stabilization gain should be (see Proposition~\ref{prop:whitening-heter-ggd-nois})
$$
20 \log_{10} \|\bs x - \bs x^*\|/\|\bs x - \bs
x^*_{\rm st}\|\ \simeq_M\ 20 \log_{10} (\epsilon \|\bs
w\|)/(\epsilon_{\rm st}\sqrt{M})\ < 2.43\, \text{dB}.
$$

Numerically, GBPDN$(\ell_{2,\bs w})$ and GBPDN$(\ell_{2})\equiv\ $BPDN have
been solved with the method described in
Section~\ref{sec:whitening-example} until the relative $\ell_2$-change in
the iterates was smaller than $10^{-6}$ (with a maximum of 2000
iterations). Reconstruction results were averaged over 50 experiments.
In Fig.~\ref{fig:rec-GBPDN-whiten-snr}, the reconstruction signal-to-noise ratio (SNR) of the
stabilized reconstruction is clearly superior to the unstabilized
one and this gain increases with increasing oversampling ratio $M/K$. This SNR gain is displayed
in Fig.~\ref{fig:rec-GBPDN-whiten-snrgain}. The dashed horizontal line
represents the theoretical prediction of $2.43$ dB which turns to be an
upper-bound on the numerically observed gain.

\begin{figure}[!htb]
  \centering
  \subfigure[\label{fig:rec-GBPDN-whiten-snr}]
  {
    \includegraphics[height=5.5cm]{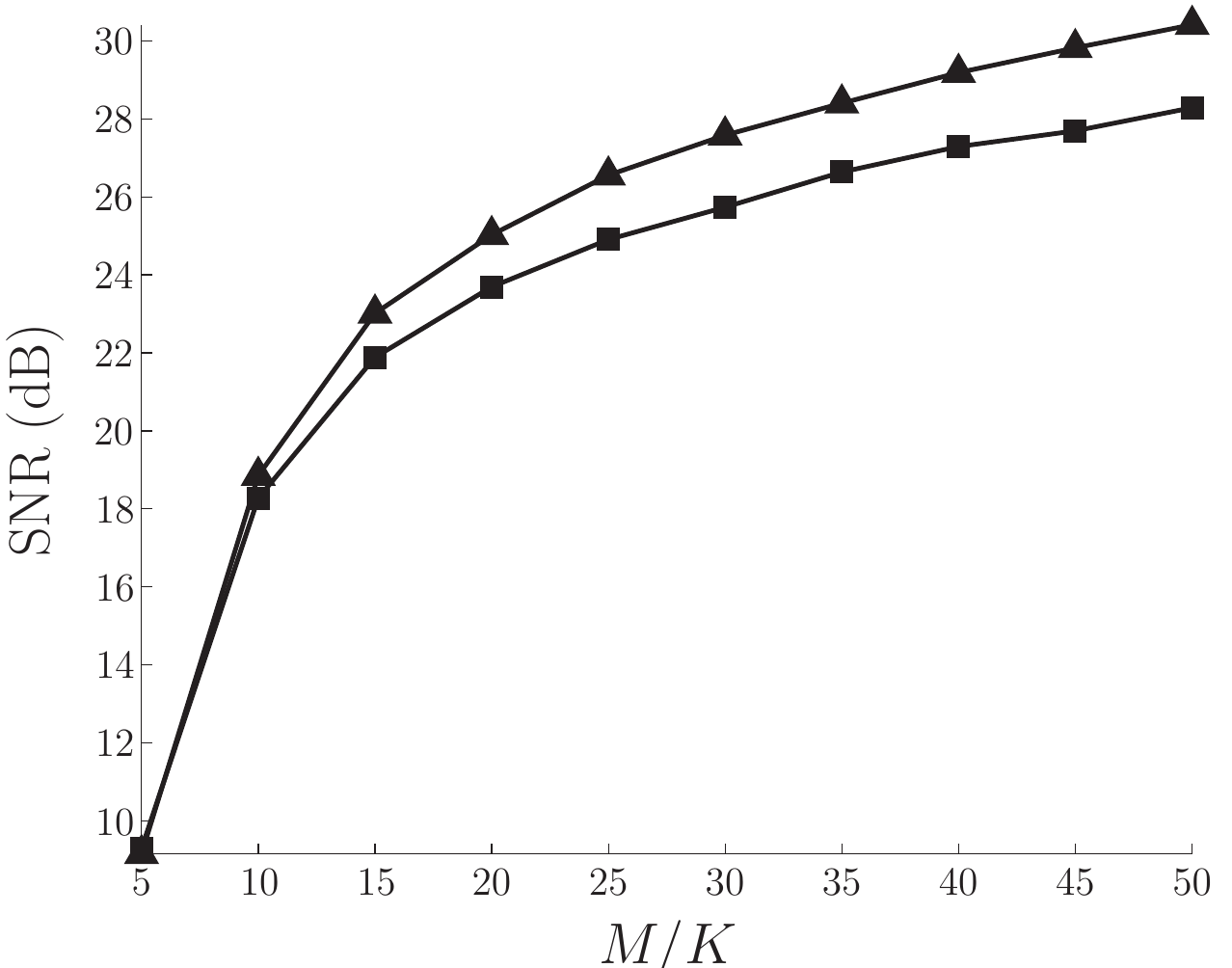}
  }
  \subfigure[\label{fig:rec-GBPDN-whiten-snrgain}]
  {
    \includegraphics[height=5.5cm]{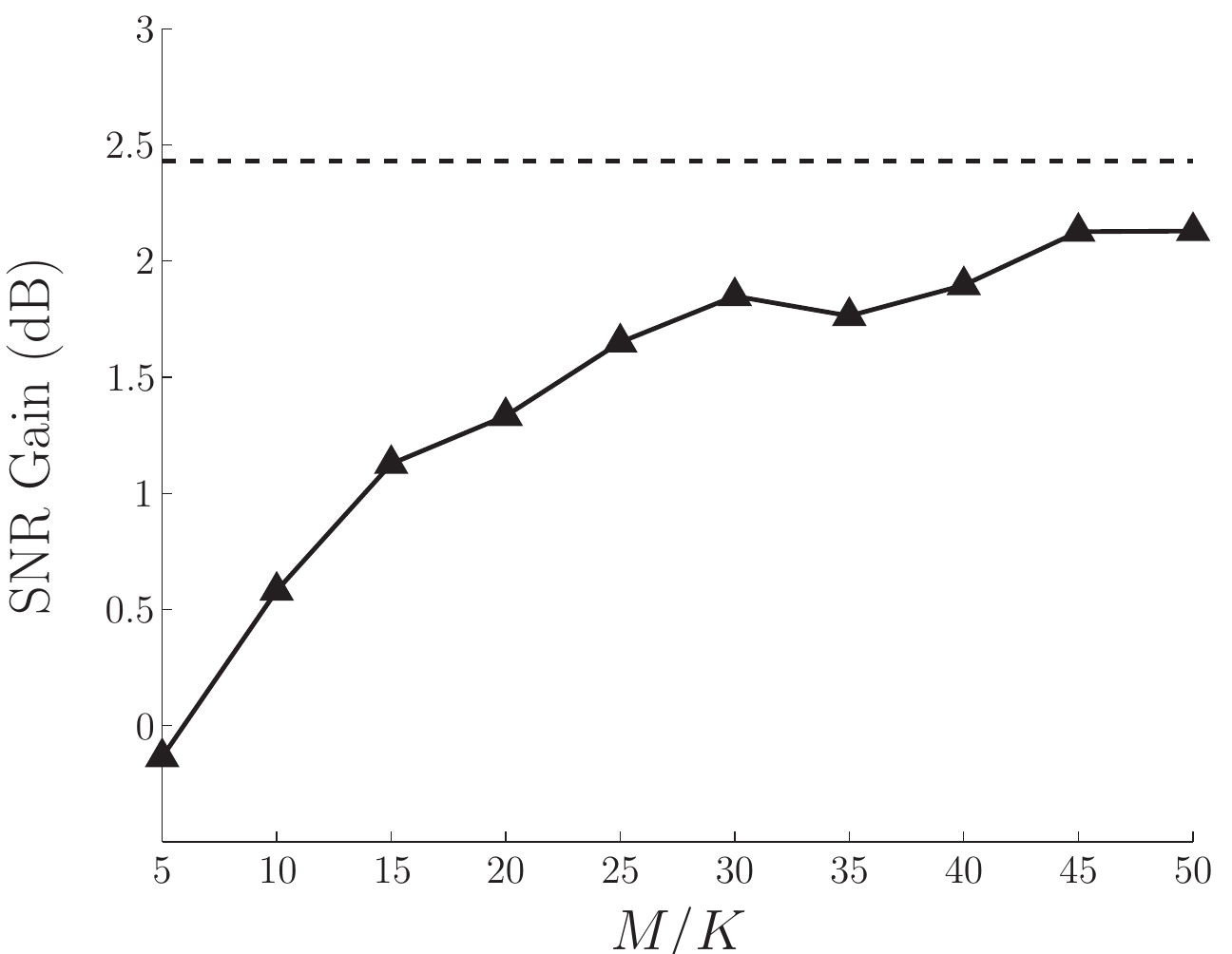}
  }
  \caption{Stabilized versus unstabilized reconstruction using
    GBPDN$(\ell_{2,\bs w})$ and BPDN respectively. (a) The reconstruction SNR using
    stabilized (triangles) and unstabilized (squares) methods. (b) Observed (triangles)
    and theoretically predicted (dashed) SNR gain at $2.43$ dB brought by 
    stabilization.}
  \label{fig:rec-GBPDN-whiten}
\end{figure}

\subsection{Non-Uniform Quantization}
\label{sec:non-unif-quant}

We describe several simulations challenging the power of GBPDN for
reconstructing sparse signals from non-uniformly quantized
measurements when the weights and the $p$-optimal levels of
Lemma~\ref{lemma:bound-lpw-gaussian-vector} are combined. Several
configurations have been tested for different $p\geq 2$, oversampling
ratio $M/K$, number of bits $B$ and for non-uniform and uniform
quantization.

For this experiment, we set the key dimensions to $N=1024, K=16, B=4$,
and the $K$-sparse unit norm signals have been generated as in the
previous section.
The oversampling ratio was taken as $M/K \in \{10,15,\,\cdots,45\}$,
$p\in \{2,4,\cdots,10\}$ and the matrix $\bs \Phi$ has
been drawn randomly as $\bs \Phi \sim \cl N^{M\times N}(0,1)$. The
non-uniform quantization of the measurements $\bs \Phi \bs x$ was
defined with a compressor $\cl G$ associated to $\gamma_{0,\sigma_0}$
according to \eqref{eq:optimal-compressor}. The weights $\bs w$ were
computed as in Lemma~\ref{lemma:bound-lpw-gaussian-vector}, and the
$p$-optimal levels using the numerical method described in
Appendix~\ref{sec:comp-omeg-p}.

For the sake of completeness, we also compared some results to those
obtained for a uniformly quantized CS scenario. In this case, the
measurements $\bs z = \bs \Phi \bs x$ are quantized as $y_i = \binw'
\lfloor z_i / \binw' \rfloor + \binw'/2$, the quantization bin width
$\binw'=\binw'(B)$ has been set by dividing regularly the interval
$[-\|\bs z\|_\infty, \|\bs z\|_\infty]$ into the same number of bins
as those used for the non-uniform quantization.

Again, GBPDN was solved with the primal-dual
scheme described in Section~\ref{sec:whitening-example} until either
the relative $\ell_2$-change in iterates was smaller than
$10^{-6}$ or a maximum number of iterations of 2000 was
reached. Finally, all the reconstruction results were averaged over 50
replications of sparse signals for each combination of parameters.

Fig.~\ref{fig:rec-GBPDN-nu-quant-snr} displays the evolution of the
signal reconstruction quality, as measured by the SNR, as a function
of the oversampling factor $M/K$. We clearly see a reconstruction
quality improvement with respect to both the uniformly quantized CS
scheme (dashed curve) and to increasing values of $p$ and $M/K$. This
last effect is better analyzed in
Fig.~\ref{fig:rec-GBPDN-nu-quant-snrgain} where the SNR gain with
respect to $p=2$ for various values of $p$ is shown. As predicted by
Proposition~\ref{prop:towards-quant-cons}, we clearly see that, as
soon as the ratio $M/K$ is large, taking higher $p$ value leads to a
higher reconstruction quality than the one obtained for $p=2$
(BPDN). Moreover, Fig.~\ref{fig:rec-GBPDN-nu-quant-snrgain} confirms
that when $p$ increases, the minimal measurement number inducing a
positive SNR gain increases. For instance, to achieve a positive gain
at $p=4$, we must have $M/K \geq 15$, while at $p=10$, $M/K$ must be
higher than 20. At $p$ fixed, the reconstruction quality increased
also monotonically with~$M/K$.

\begin{figure}[!Htb]
  \centering
  \subfigure[\label{fig:rec-GBPDN-nu-quant-snr}]
  {
    \includegraphics[height=5.5cm]{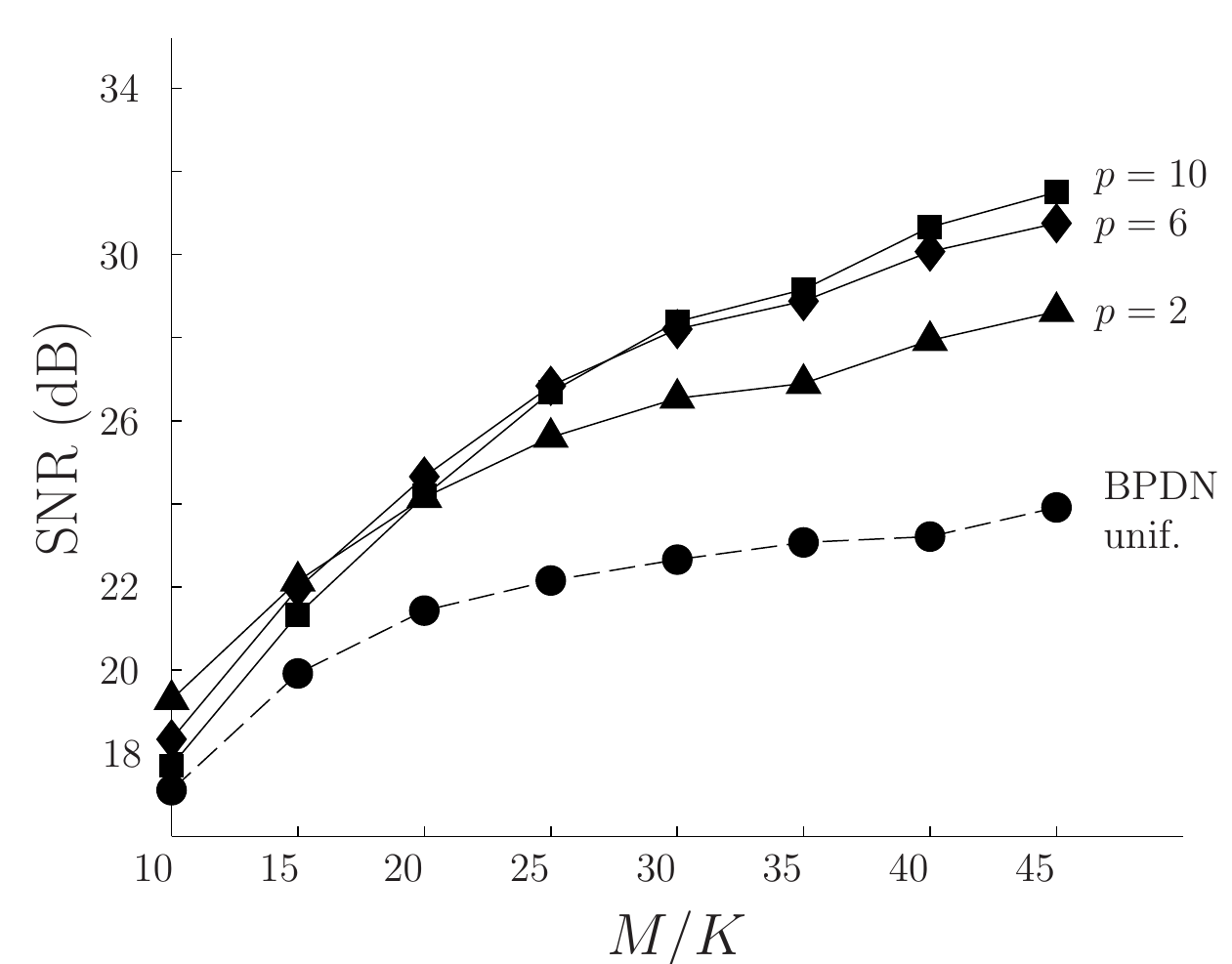}
  }
  \subfigure[\label{fig:rec-GBPDN-nu-quant-snrgain}]
  {
    \includegraphics[height=5.5cm]{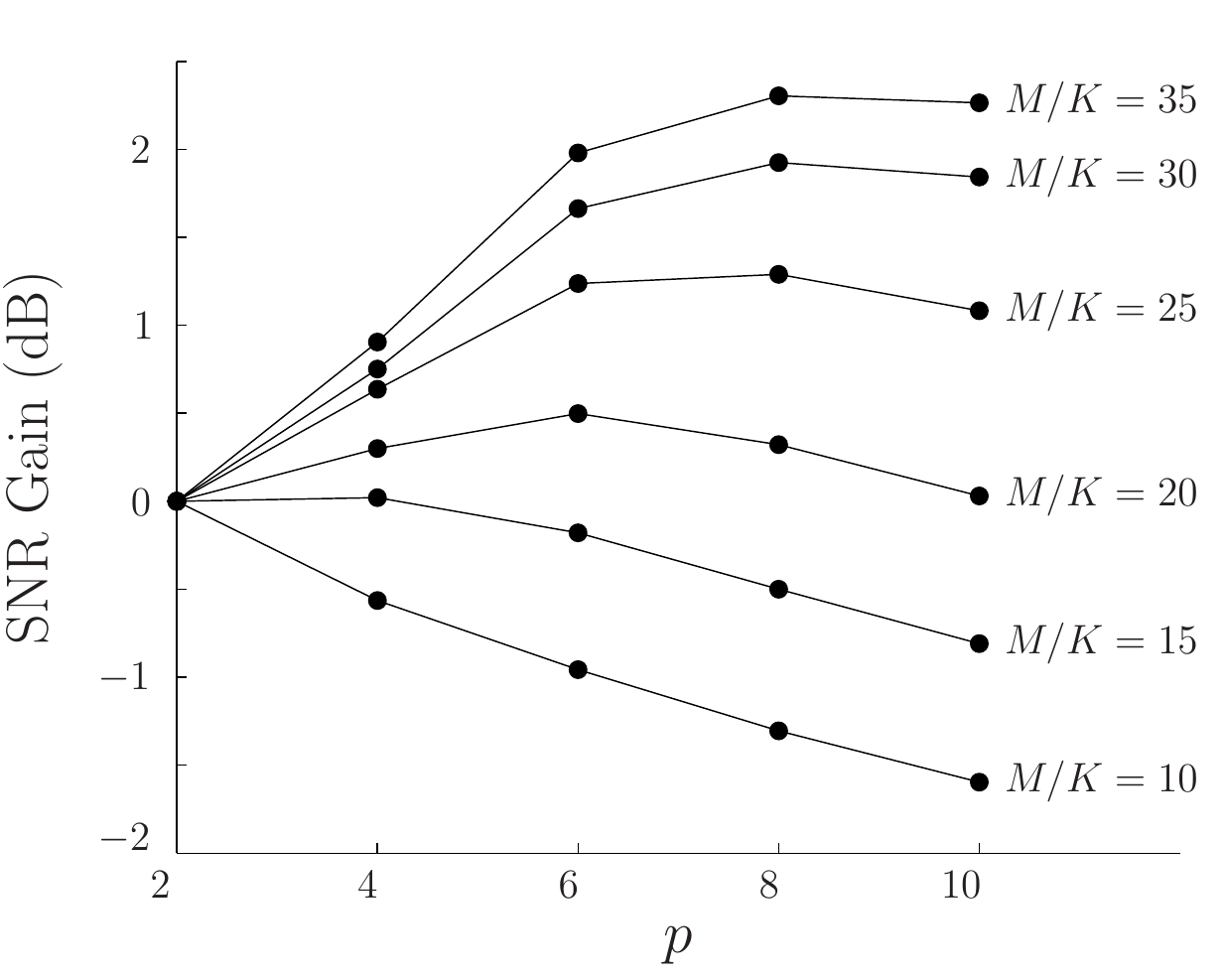}
  }
  \caption{Reconstruction SNR of GBPDN$(\ell_{p,\bs w})$. (a): the dashed line represents
    the reconstruction quality achieved from uniformly quantized CS
    and BPDN. (b) SNR gain versus $p$ for each tested oversampling ratio $M/K$.}
  \label{fig:rec-GBPDN-nu-quant}
\end{figure}

We observe that, given the oversampling ratio, these experimental results allow to increase $p$ to a
greater extent than would be allowed by our theory deployed in Section~\ref{sec:dequ-with-gener}. In particular, the sufficient condition \eqref{eq:SGR-RIP-measur-bound} dictated by Proposition~\ref{prop:grip-gauss} requires the number of measurements $M$ to scale as $K^{p/2}$ (ignoring times the usual logarithmic terms) in order to ensure
the RIP$_{p,\bs w}$. This would imply an exponential increase in the number of measurements needed as $p$ increases. However, from Fig.~\ref{fig:rec-GBPDN-nu-quant-snrgain}, one can see that for $M/K=15$,
$p=4$ was the largest value before performance
starts degrading. With $M/K=20$, $p$ could be increased to 6 before
degradation, and to 8 before degradation with $M/K=30$. At least for
this example, we do not observe such a severe exponential dependence in the
needed oversampling in order to benefit from error decrease when
increasing $p$.

In Fig.~\ref{fig:GBPDN-QC}, the quantization consistency
of the reconstructed signals is tested by looking at the histogram of
$\binw^{-1} (\cl G(\bs\Phi\bs x^*) - \cl G(\bs y))$. We do observe
that this histogram is closer to a uniform distribution for $p=10$
than for $p=2$, in good agreement with the ``companded'' quantizer
definition $\cl Q = \cl G^{-1}\circ\cl Q_\binw \circ \cl G$ showing
that in the domain compressed by $\cl G$, this quantizer is similar to
a uniform one.

\begin{figure}[!Htb]
  \centering
  \subfigure[\label{fig:GBPDN-QC-p2}]
  {
    \includegraphics[height=4cm]{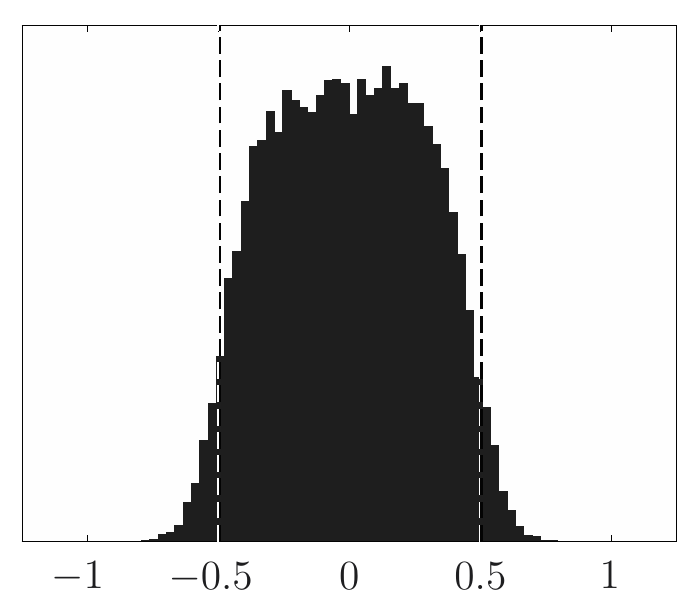}
  }
  \subfigure[\label{fig:GBPDN-QC-p10}]
  {
    \includegraphics[height=4cm]{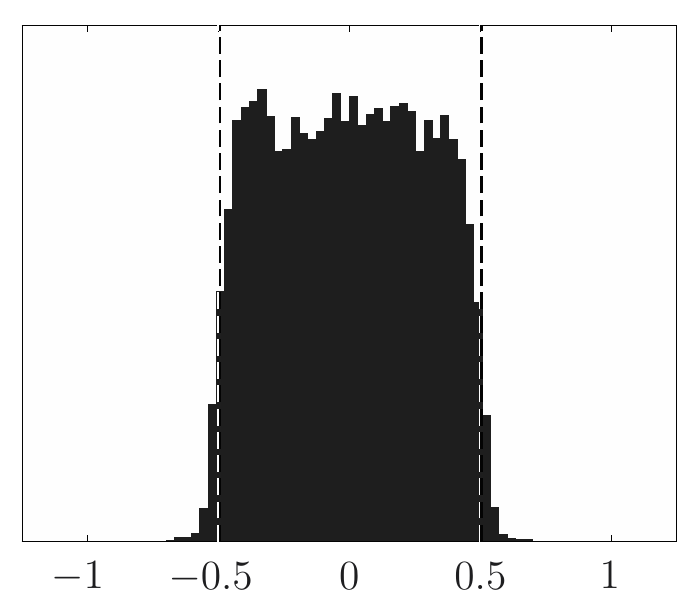}
  }
  \caption{Testing the Quantization Consistency (QC). (a) Histogram of
    the components of
    $\binw^{-1}\,\big(\cl G(\bs \Phi \bs x^*) - \cl G(\bs y)\big)$
    for $p=2$ and $M/K=40$ (averaged over 100 trials). (b) Same histogram for $p=10$. The QC is
    better respected in this case.
    \label{fig:GBPDN-QC}
  }
\end{figure}

As a last test, we have more thoroughly compared a uniform
quantization scenario described in the experimental setup above with
the BPDQ$_p$ decoder developed in~\cite{Jacques2010} to the
non-uniform case studied in this paper. More precisely,
Fig.~\ref{fig:rec-GBPDN-nu-quant-exp-vs-th} shows the reconstruction
SNR gain between non-uniform and uniform quantization at various $p$,
\ie SNR$($GBPDN$(\ell_{p,\bs w}))$ $-$ SNR$($BPDQ$_p)$. We see that, at
a given $p$, this gain improves with $M/K$, and the highest SNR improvement values
are obtained for $p=2$. This points the fact that for $p\neq 2$, the
quantization scheme is not optimized for reducing the $\ell_{p,\bs
  w}$-norm distortion. This would require us to change the
quantization scenario by not only optimizing the $p$-optimal levels
but also the thresholds. This will be be left to a future research.

\begin{figure}[!Htb]
  \centering
  {
    \includegraphics[height=5.5cm]{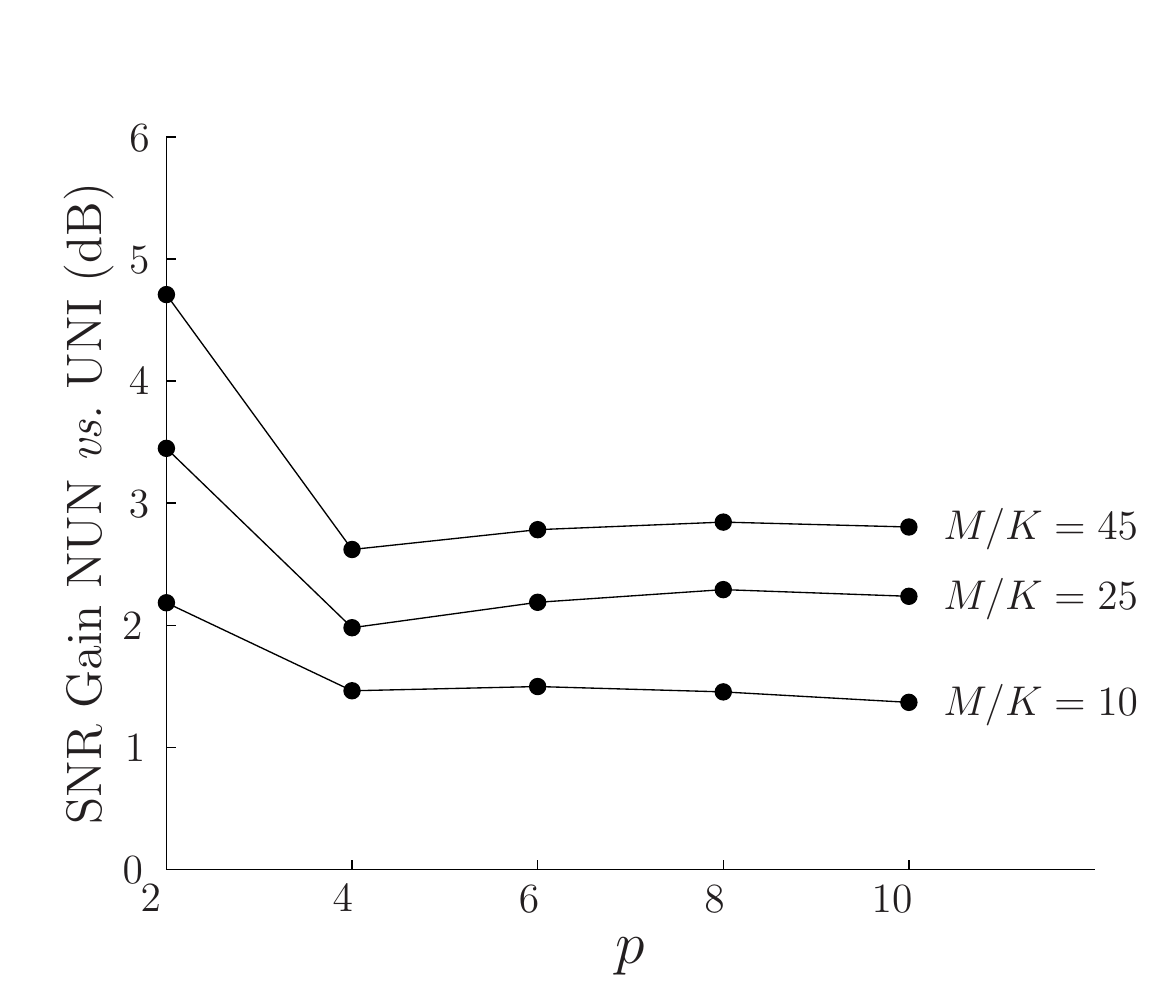}
  }
  \caption{\label{fig:rec-GBPDN-nu-quant-exp-vs-th} 
Reconstruction gain (in dB) between non-uniform or uniform
    quantization at the same $p$.}
\end{figure}

\section{Conclusion}
\label{sec:conclusion}

In this paper, we have shown that, when the compressive measurements
of a sparse or compressible signal are non-uniformly quantized, there
is a clear interest in modifying the reconstruction procedure by
adapting the way it imposes the reconstructed signal to ``match'' the
observed data. In particular, we have proved that in an oversampled
scenario, replacing the common BPDN $\ell_2$-norm constraint by a
weighted $\ell_p$-norm adjusted to the non-uniform nature of the
quantizer reduces the reconstruction error by a factor of
$\sqrt{p+1}$. Moreover, we showed that this improvement stems from a
stabilization of the quantization distortion seen as an additive
heteroscedastic GGD noise on the measurements.

In  future work, we will investigate if the quantization scheme can
also be optimized with respect to the proposed reconstruction
procedure, \ie by adjusting the thresholds for minimizing the weighted
$\ell_p$-distortion at a fixed bit budget.  

\appendices

\section{Preparatory  Lemmata}
\label{sec:aux-lemmata}

This appendix contains several key lemmata that are useful for the
subsequent proofs developed in the other appendices.

The first lemma will serve later to evaluate asymptotically the
contribution of each quantization bin to the global quantizer
distortion measured with $\ell_{p,\bs w}$-norm when a Gaussian source
(with pdf $\pdf_0$) is quantized.

\begin{lemma}
\label{lemma:optimal-p-moment-bound}
Given $a,b\in\Rbb$ with $a<b$, $n\in\Nbb\setminus\{0\}$ and a Gaussian pdf $\pdf_0=\gamma_{0,\sigma_0}$. Let
$\lambda_n$ be the (unique) minimizer of $\min_{\lambda\in[a,b]}\ \int_a^b |t-\lambda|^n\ \pdf_0(t)\,\ud t$.
Then, 
\begin{align}
\label{eq:lower-bound-intab}
\int_a^b |t-\lambda_n|^n\ \pdf_0(t)\,\ud t\ &\geq\ \tfrac{(b -
a)^{n+1}}{(n+1)\,2^{n+1}}\,(1+ (\tfrac{\fs D}{\fs C})^{-(n+1)/n})\,{\fs C},\\ 
\label{eq:upper-bound-intab}
\int_a^b |t-\lambda_n|^n\ \pdf_0(t)\,\ud t\ &\leq\ \tfrac{(b -
a)^{n+1}}{(n+1)\,2^{n+1}}\,(1+ (\tfrac{\fs D}{\fs C})^{(n+1)/n})\,{\fs
D},
\end{align}\vspace{-6mm}
\begin{align}
\label{eq:behav-lambda-finite-bin}
\tfrac{1}{1+\fs S^{1/n}}(\fs S^{1/n}a + b)&\ \leq\ 
\lambda_n\ \leq\ \tfrac{1}{1+\fs S^{1/n}}(a + \fs S^{1/n}b),
\end{align}
with $\fs C\eqdef\min_{t\in[a,b]}\pdf_0(t)$, $\fs
D\eqdef\max_{t\in[a,b]}\pdf_0(t)$ and $\fs S=\fs D/\fs C$. 
\end{lemma}
\begin{proof}
  Let us first show the upper bound \eqref{eq:upper-bound-intab}.  In
  Lemma~\ref{lemma:p-optimal-levels} and its proof, it was show that
  $\lambda_n$ exists and is unique, \ie the minimization problem is
  well-posed. Furthermore, $\lambda_n$ satisfies
\begin{equation*}
\fs A := \int_a^{\lambda_n} (\lambda_n - t)^{n-1}\ \pdf_0(t)\,\ud t\ =\ \int_{\lambda_n}^b (t-\lambda_n)^{n-1}\ \pdf_0(t)\,\ud t.
\end{equation*}
Since $\pdf_0(t)\in [\fs C, \fs D]$ for $t\in[a,b]$, we have 
$(\lambda_n - a)^n\,\fs C \leq n\fs A \leq (\lambda_n
- a)^n\,\fs D$ and $(b - \lambda_n)^n\,\fs C \leq n\fs A \leq
\tinv{n}(b - \lambda_n)^n\,\fs D$.  
This implies $(\lambda_n - a)^n \geq
\big(\tfrac{\fs C}{\fs D}\big) (b -
\lambda_n)^n$ and $(b - \lambda_n)^n \geq
\big(\tfrac{\fs C}{\fs D}\big)
(\lambda_n - a)^n$, from which we easily
deduce~\eqref{eq:behav-lambda-finite-bin}.

Since 
$\int_a^b |t-\lambda_n|^n\ \pdf_0(t)\,\ud t = \int_a^{\lambda_n}
(\lambda_n-t)^n\ \pdf_0(t)\,\ud t\ + \int_{\lambda_n}^b
(t-\lambda_n)^n\ \pdf_0(t)\,\ud t$,
we find $\int_a^b |t-\lambda_n|^n\ \pdf_0(t)\,\ud t  \leq \tinv{n+1}\,[(\lambda_n - a)^{n+1} + (b - \lambda_n)^{n+1}]\,\fs D$.
From $((\lambda_n-a)/(b-\lambda_n))^n\in [\fs C/\fs D, \fs D/\fs C]$, we find that $\int_a^b
|t-\lambda_n|^n\ \pdf_0(t)\,\ud t$ is smaller than 
$$
\tinv{n+1}\min\big(\,(\lambda_n - a)^{n+1}, (b - \lambda_n)^{n+1}\big) \left[1 + \big(\tfrac{\fs
    D}{\fs C}\big)^{(n+1)/n}\right]\, \fs D.
$$
This provides \eqref{eq:upper-bound-intab} since $\min(\lambda_n-a,b-\lambda_n)\leq (b-a)/2$.
The bound \eqref{eq:lower-bound-intab} is obtained similarly. 
\end{proof}

The following lemma presents a generalization of ``$Q$-function like''
bounds for lower partial moments of a Gaussian pdf. 

\begin{lemma}
\label{lemma:bounding-lpm}
Let $\lambda > 0$, $n\in\Nbb$ and $\varphi = \gamma_{0,1}$. Let us
define  $Q_n(\lambda) \eqdef \int_\lambda^{+\infty}\ (t - \lambda)^n
\varphi(t)\ \ud t$. Then, $Q_n(\lambda) = \Theta(\lambda^{-(n+1)}\varphi(\lambda))$. More precisely,
$
\tfrac{n!\,\lambda^{n+1}}{\Pi_{k=1}^{n+1}(\lambda^2 +
k)}\,\varphi(\lambda)\ \leq\ Q_n(\lambda)\ \leq\ \tfrac{n!}{\lambda^{n+1}}\,\varphi(\lambda). 
$
\end{lemma}
This lemma generalizes the well known bound on $Q=Q_0$, namely $\tfrac{\lambda}{\lambda^2 +
1}\,\varphi(\lambda) \leq Q(\lambda) \leq \tfrac{1}{\lambda}\,\varphi(\lambda)$.
\begin{proof}
The proof involves integration by parts, the identities $-\varphi'(u) = u\varphi(u)$ and
$({\varphi(u)}/{u^n})' = (1+ \tfrac{n}{u^2})
\tfrac{\varphi(u)}{u^{n-1}}$.  
Therefore, the upper bound is a simple consequence of
$$
Q_n(\lambda) \leq \tinv{\lambda}\int_{\lambda}^{+\infty} (t-\lambda)^n
\ t\varphi(t)\ \ud t = \tfrac{n}{\lambda}\, Q_{n-1}(\lambda)\leq
\cdots \leq \tfrac{n!}{\lambda^n} Q(\lambda) \leq \tfrac{n!}{\lambda^{n+1}} \varphi(\lambda). 
$$

To get the lower bound, observe first that, defining $Q_{n,k}(\lambda) \eqdef \int_\lambda^{+\infty}\ (t - \lambda)^n
t^{-k}\varphi(t)\ \ud t$, we find
$$
(1 + \tfrac{k+1}{\lambda^2})\,Q_{n,k}(\lambda) \geq \int_\lambda^{+\infty}\ (t - \lambda)^n
(1 + \tfrac{k+1}{t^2})\,t^{-k}\varphi(t)\ \ud t = n\,Q_{n-1,k+1}(\lambda).
$$
Therefore, $Q_n(\lambda) \geq
\tfrac{n\lambda^2}{\lambda^2+1}\,Q_{n-1,1}(\lambda)\geq\cdots\geq \tfrac{n!\lambda^{2n}}{\Pi_{k=1}^n(\lambda^2+k)}\,Q_{0,n}(\lambda)$.
But $(1 + \tfrac{n+1}{\lambda^2})\,Q_{0,n}(\lambda) \geq
{\varphi(\lambda)}/{\lambda^{n+1}}$, so that
$
Q_n(\lambda) \geq \tfrac{n!\lambda^{2n+2}}{\Pi_{k=1}^{n+1}(\lambda^2+k)}\,\tfrac{\varphi(\lambda)}{\lambda^{n+1}}$, which concludes the proof.
\end{proof}

\section{Proof of Lemma~\ref{lemma:p-optimal-levels}: ``$p$-optimal Level Definiteness''}
\label{sec:p-optimal-level-definiteness}

\begin{proof} 
  For $2\leq p < \infty$, $|t - \lambda|^p$ is a continuous, coercive and strictly convex function of $\lambda$ over $\Rbb$, and therefore so is $\int_{\cl R_k} |t - \lambda|^p\ \pdf_0(t)\,\ud t$ since $\pdf_0(t)>0$.
  It follows that the function $\int_{\cl R_k} |t -
  \lambda|^p\ \pdf_0(t)\,\ud t$ has a unique minimizer on $\Rbb$. Moreover, this minimizer is necessarily located in $\cl R_k$ since $\int_{\cl R_k} |t - \lambda|^p\ \pdf_0(t)\,\ud t$ is monotonically decreasing (resp. increasing) on $(-\infty,t_k)$ (resp. $(t_{k+1},+\infty)$)\footnote{Where we used the Lebesgue dominated convergence theorem to interchange the integration and derivation signs.}. Consequently, $\omega_{k,n}$ exists and is unique.

For proving the limit case $p\to\infty$, for finite bins $\cl R_k$
($k\notin \{1,\cl B\}$) and without loss of generality for $t_k\geq 0$, relation \eqref{eq:behav-lambda-finite-bin}
in Lemma~\ref{lemma:optimal-p-moment-bound} with $a=t_k$ and
$b=t_{k+1}$, together with the squeeze theorem shows that 
$$
\lim_{p\to+\infty} \omega_{k,p} =
\lim_{p\to+\infty} \tfrac{1}{1+\fs S^{1/p}}(\fs S^{1/p}t_k + t_{k+1}) =
\lim_{p\to+\infty} \tfrac{1}{1+\fs S^{1/p}}(t_k + \fs S^{1/p} t_{k+1}) = \omega_{k,\infty} ~,
$$
where $\fs S=\pdf_0(t_k)/\pdf_0(t_{k+1})$. 

For infinite bins (\ie $k\in\{1,\cl B\}$) and assuming again $t_k\geq 0$, it follows from the beginning of the proof that $\omega_{k,p}$ is the unique root on $[t_k,+\infty)$ of $\cl E_p(\lambda) \eqdef
\int_{t_k}^\lambda (\lambda - t)^{p-1}\pdf_0(t)\,\ud t -
\int_{\lambda}^\infty (t-\lambda)^{p-1}\pdf_0(t)\,\ud t$.  Let
$\tilde\omega_{k,p} \in [t_k,L]$ be the root of $\tilde{\cl E}_p(\lambda,L) \eqdef
\int_{t_k}^\lambda (\lambda - t)^{p-1}\pdf_0(t)\,\ud t -
\int_{\lambda}^L (t-\lambda)^{p-1}\pdf_0(t)\,\ud t$ for some $L\geq
t_k$.  We then have $\cl E_p(\tilde\omega_{k,p}) =
\int_{t_k}^{\tilde\omega_{k,p}} ({\tilde\omega_{k,p}} -
t)^{p-1}\pdf_0(t)\,\ud t - \int_{{\tilde\omega_{k,p}}}^\infty
(t-{\tilde\omega_{k,p}})^{p-1}\pdf_0(t)\,\ud t = -\int_{L}^\infty
(t-{\tilde\omega_{k,p}})^{p-1}\pdf_0(t) \leq 0 = \cl
E_p(\omega_{k,p})$, which implies $\tilde\omega_{k,p} \leq
\omega_{k,p}$ since $\cl E_p$ is non-decreasing for $p\geq 1$. However,
since $\tilde\omega_{k,p}$ is optimal on $[t_k,L]$, taking
$L=L(p)=c\sqrt p$, for $c>0$, we have by Lemma~\ref{lemma:optimal-p-moment-bound} with $a=t_k$ and
$b=L(p)$, $\lim_{p\to +\infty} \tilde\omega_{k,p} \geq
\lim_{p\to +\infty}\tfrac{1}{1+\fs S^{1/p}}(\fs S^{1/p}t_k + c\sqrt p)
= +\infty$
since $\fs S^{1/p} =
\exp(-t^2_k/2p\sigma_0^2)\exp(c^2/2\sigma_0^2)=\Theta(1)$. This proves
$\lim_{p\to +\infty} |\omega_{k,p}| = +\infty = \omega_{k,\infty}$ and
$|\omega_{k,p}|=\Omega(\sqrt p)$ for $k\in\{1,\cl B\}$.
\end{proof}

\section{Proof of Lemma~\ref{lemma:useful-ineq}: ``Asymptotic $p$-Quantization Characterization''}
\label{proof-toolbox-lemma}

The content of Lemma \ref{lemma:useful-ineq} is derived from this
larger set of results which constitutes a \emph{toolbox} lemma for
other developments given in these appendices.  
\begin{lemma}[\bf Extended Asymptotic $p$-Quantization Characterization]
\label{lemma:useful-ineq-ext}
Given the Gaussian pdf $\pdf_0$ and its associated compressor $\cl G$
function, choose $0<\beta<1$ and $p\in\Nbb$, and define
$T=T(B)=\sqrt{6\,\sigma_0^2(\log 2^\beta)\,B}$, $\cl T=[-T,T]$ and $\cl
T^c = \Rbb \setminus \cl T$. We have the following asymptotic
properties (relative to $B$):
\begin{align}
\label{lemma:useful-ineq-ext-eq1}
\cl G'(T(B))&= \Theta(2^{-\beta B}),\\
\label{lemma:useful-ineq-ext-eq2}
\#\{k : \cl R_k \subset \cl T^c\}&= \Theta\big(B^{-1/2}\,2^{(1-\beta)B}\big),\\
\label{lemma:useful-ineq-ext-eq2bis}
\int_{\cl R_k} |t-\omega_{k,p}|^p\
\pdf_0(t)\,\ud t&= O\big(B^{-(p+1)/2}\,2^{-3\beta B}\big),\quad \forall 
\cl R_k \subset \cl T^c.
\end{align}
Moreover, for all $k$ such that $\cl R_k \subset \cl T$ and any $c\in
\cl R_k$
\begin{eqnarray}
\label{lemma:useful-ineq-ext-eq3}
&\gbinw_k\eqdef t_{k+1}-t_k =\ O(2^{-(1-\beta)B}),\\
\label{lemma:useful-ineq-ext-eq4}
&1\ \leq\
\tfrac{\max(\pdf_0(t_k),\,\pdf_0(t_{k+1}))}{\min(\pdf_0(t_k),\,\pdf_0(t_{k+1}))}
=\ \exp\big(O(B^{1/2}\,
2^{-(1-\beta)B})\big)= 1 + O(B^{1/2}\,
2^{-(1-\beta)B}),\\
\label{lemma:useful-ineq-ext-eq7}
&\textstyle \int_{\cl R_k} |t-\omega_{k,p}|^p\ \pdf_0(t)\,\ud t\ \simeq_B\
\tfrac{\gbinw_k^{p+1}}{(p+1)\,2^{p}}\,\pdf_0(c),\\
\label{lemma:useful-ineq-ext-eq8}
&\cl G'(c) \simeq_B \tfrac{\binw}{\gbinw_k}.
\end{eqnarray}
Finally, if $k$ is such that $T(B) \in \cl R_k$, then, writing the
interval length/measure $\cl L(\cl A)=\int_{\cl A} \ud t$ for $\cl A\subset \Rbb$,
\begin{align}
\label{lemma:useful-ineq-ext-eq9}
\cl L(\cl R_k \cap \cl T)&=\ O(2^{-(1-\beta)B}),\\
\label{lemma:useful-ineq-ext-eq10}
\cl G'(\omega_{k,p})&\leq \max(\cl G'(t_k),\cl G'(t_{k+1})) =\ O(2^{-\beta B}),\\
\label{lemma:useful-ineq-ext-eq11}
\int_{\cl R_k} |t-\omega_{k,p}|^p\ \pdf_0(t)\,\ud t\ &= O\big(B^{-(p+1)/2}\,2^{-3\beta B}\big).
\end{align}
\end{lemma}

\begin{proof}
In this proof we use the quantizer symmetry to restrict the analysis
to the half (positive) real line $\Rbb_+$, on which $\pdf_0$ is decreasing.

Relation \eqref{lemma:useful-ineq-ext-eq1} comes from the definition
of $T(B)$ and that of $\cl G'=\gamma_{0,\sqrt 3\,\sigma_0}$. For proving
\eqref{lemma:useful-ineq-ext-eq2}, we can observe that $\cl G(\lambda) =
\fnorm{\pdf_0}^{-1/3}_{1/3} \int_{-\infty}^\lambda
\pdf_0^{1/3}(t)\,\ud t = 1-Q(\lambda/\sqrt 3 \sigma_0)$ where $Q(t) =
\tinv{\sqrt{2\pi}}\int_t^{+\infty} \gamma_{0,1}(u)\,\ud u$.  Since
$\frac{\lambda}{1+\lambda^2}\gamma_{0,1}(\lambda)
\leq Q(\lambda) \leq \tfrac{1}{\lambda}
\gamma_{0,1}(\lambda)$, we obtain
$$
\tfrac{3 \sigma^2_0 \lambda}{3\sigma_0^2+\lambda^2}\,\cl G'(\lambda)\ \leq\ 1 - \cl G(\lambda)\ \leq\ \tfrac{3\sigma_0^2}{\lambda}
\,\cl G'(\lambda).  
$$
Taking $\lambda=T(B)$ in the last inequalities and using \eqref{lemma:useful-ineq-ext-eq1}, we deduce from the quantizer definition
$$
\#\{k : \cl R_k \subset \cl T^c\}\ =\ 2\,\#\{k : t_k \geq T(B)\}\ =\ 2\,\binw^{-1}\,(1-\cl G(T))\ = \Theta\big(B^{-1/2}\,2^{(1-\beta)B}\big).
$$
Relation \eqref{lemma:useful-ineq-ext-eq2bis} is proved by noting that,
if $t_k \geq T(B)$, 
$$
\int_{\cl R_k} |t-\omega_{k,p}|^p\
\pdf_0(t)\,\ud t \leq \int_{\cl R_k} (t-t_k)^p\
\pdf_0(t)\,\ud t \leq \int_{t_k}^{\infty} (t-t_k)^p\
\pdf_0(t)\,\ud t,
$$
where the first inequality follows from the $p$-optimality of $\omega_{k,p}\in\cl R_k$.  However, from
Lemma~\ref{lemma:bounding-lpm}, we know that, for $\lambda\in\Rbb_+$
\begin{equation*}
\tfrac{p!\,\lambda^{p+1}\sigma_0^{2p+2}}{\Pi_{k=1}^{p+1}(\lambda^2 +
k\sigma_0^2)}\,\pdf_0(\lambda)\ \leq\ \sigma_0^{p}\,Q_{p}(\tfrac{\lambda}{\sigma_0})\ \leq\ \tfrac{p!\,\sigma_0^{2p+2}}{\lambda^{p+1}}\,\pdf_0(\lambda),  
\end{equation*}
with $Q_{p}(\lambda) \eqdef \int_{\lambda}^{\infty}\,
(t-\lambda)^{p}\,\gamma_{0,1}(t)\ \ud t$ and $\sigma_0^{p}\,Q_{p}(\tfrac{\lambda}{\sigma_0}) = \int_{\lambda}^{\infty}\,
(t-\lambda)^{p}\,\varphi_0(t)\ \ud t$.

Therefore, since $\pdf_0 \propto (\cl G')^3$,
$$
\int_{t_k}^{\infty} (t-t_k)^p\
\pdf_0(t)\,\ud t \leq\
\tfrac{p!\,\sigma_0^{2(p+1)}}{t_k^{p+1}}\,\pdf_0(t_k) \leq
\tfrac{p!\,\sigma_0^{2(p+1)}}{T^{p+1}}\,\pdf_0(T) = O\big(B^{-(p+1)/2}\,2^{-3\beta B}\big).
$$

Relation \eqref{lemma:useful-ineq-ext-eq3} is obtained by observing
that $\cl G$ is concave on $\Rbb_+$. This implies $\gbinw_k \leq
\binw/\cl G'(t_{k+1})$ and if $k$ is such that $0 \leq t_{k+1} \leq
T(B)$, $\gbinw_k = O(2^{-(1-\beta)B})$. For
\eqref{lemma:useful-ineq-ext-eq4}, keeping the same $k$, we note that $1
\leq \frac{\pdf_0(t_k)}{\pdf_0(t_{k+1})} =
\exp(\inv{6\sigma_0^2}\gbinw_k (t_k+t_{k+1})) \leq
\exp(\inv{3\sigma_0^2}\gbinw_k t_{k+1}) = \exp\big
(O(B^{1/2}\,2^{-(1-\beta)B})\big)$ which is then arbitrarily close
to~1.

For proving \eqref{lemma:useful-ineq-ext-eq7}, we assume first $p\geq 1$.
Let us consider \eqref{eq:lower-bound-intab} and
\eqref{eq:upper-bound-intab} with $a=t_k$, $b=t_{k+1}$, $\fs
C=\pdf_0(t_{k+1})$ and $\fs D=\pdf_0(t_{k})$ with $0\leq t_{k+1} \leq
T(B)$.  From \eqref{lemma:useful-ineq-ext-eq4} we see that $1 \leq
\tfrac{\fs D}{\fs C} = 1 + o(1)$. We show easily that this involves
the equivalent relations $\fs C \simeq_B \fs D$, $\fs C/\fs D \simeq_B
1$ and $\fs D/\fs C \simeq_B 1$. Therefore, $(1+(\fs D/\fs
C)^{(p+1)/p}) \simeq_B 2$ and $(1+(\fs C/\fs D)^{(p+1)/p}) \simeq_B
2$. Moreover, $\fs C \simeq_B \pdf_0(c)$ and $\fs D \simeq_B
\pdf_0(c)$ for any $c\in\cl R_k$, so that \eqref{eq:lower-bound-intab}
and \eqref{eq:upper-bound-intab}) show finally $\int_{\cl R_k}
|t-\omega_{k,p}|^p\ \pdf_0(t)\,\ud t \ \mathop{\lesssim}_B\
\tfrac{\gbinw_k^{p+1}}{(p+1)\,2^{p}}\,\pdf_0(c)$ and $\int_{\cl R_k}
|t-\omega_{k,p}|^p\ \pdf_0(t)\,\ud t \ \mathop{\gtrsim}_B\
\tfrac{\gbinw_k^{p+1}}{(p+1)\,2^{p}}\,\pdf_0(c)$, which proves the
relation. The case $p=0$ is demonstrated similarly by observing that
$\pdf_0(t_{k+1})\gbinw_k \leq p_k\eqdef \int_{\cl R_k}\pdf_0(t)\,\ud t
\leq \pdf_0(t_k)\gbinw_k$.

Let's now turn to showing \eqref{lemma:useful-ineq-ext-eq8}. 
From \eqref{lemma:useful-ineq-ext-eq4} and since $\cl G' \propto
\pdf_0^{1/3}$, $1 \leq \cl G'(t_{k})/\cl G'(t_{k+1}) = 1 + o(1)$ so
that $\cl G'(t_{k})/\cl G'(t_{k+1})\simeq_B 1$.  By concavity of
$\cl G$ on $\Rbb_+$, we know that $\cl G'(t_{k+1}) \leq \binw/\gbinw_k
\leq \cl G'(t_{k})$. Therefore, $1 \leq (\cl
G'(t_{k+1}))^{-1}\binw/\gbinw_k = 1 + o(1)$ which yields $\cl G'(t_{k+1})
\simeq_B \binw/\gbinw_k$.  By the concavity argument again, we have $\cl G'(t_{k})\geq \cl
G'(c) \geq \cl G'(t_{k+1})$ for any $c \in \cl R_k$, and thus $1+o(1)=\cl G'(t_{k})/\cl G'(t_{k+1})\geq \cl G'(c)/\cl G'(t_{k+1}) \geq 1$. This implies $\cl G'(c) \simeq_B \cl G'(t_{k+1}) \simeq_B \binw/\gbinw_k$.

If $k$ is such that $0 \leq t_{k} \leq T(B) \leq t_{k+1}$,
using again the concavity of $\cl G$ on $\Rbb_+$, we find $\cl L(\cl
R_k \cap \cl T) = T(B) - t_k
\leq \cl (\cl G(T(B)) - k\binw)/\cl G'(T(B)) \leq \binw/\cl G'(T(B)) =
O(2^{-(1-\beta)B})$, which proves \eqref{lemma:useful-ineq-ext-eq9}.

For showing \eqref{lemma:useful-ineq-ext-eq10}, we note that $\cl G'(t_k) =
\cl G'(T) (\cl G'(t_k)/\cl G'(T))$. Since $\cl G'(t_k)/\cl
G'(T)=\exp(\inv{6\sigma_0^2}(T-t_k)(T+t_k))\leq
\exp(\inv{3\sigma_0^2}(T-t_k)T) =
\exp(O(B^{1/2}\,2^{-(1-\beta)B}))$ which is arbitrarily close to 1
(\ie it is $e^{o(1)}$), we find $\cl G'(t_k) = O(2^{-\beta B})$, \ie
it inherits the behavior of $\cl G'(T)$. 

The last relation \eqref{lemma:useful-ineq-ext-eq11} is proved similarly to \eqref{lemma:useful-ineq-ext-eq2bis} by
appealing again to Lemma~\ref{lemma:bounding-lpm},
$$
\int_{\cl R_k} (t-t_k)^p\ \pdf_0(t)\,\ud t\ \leq\ \int_{t_k}^{\infty}
(t-t_k)^p\ \pdf_0(t)\,\ud t \leq\ 
\tfrac{p!\,\sigma_0^{2p+2}}{t_k^{p+1}}\,\pdf_0(t_k) = O\big(B^{-(p+1)/2}\,2^{-3\beta B}\big),
$$
where the asymptotic relation is obtained by seeing that, as soon as
$T-t_k \leq 1/2$ (which is always possible to meet thanks to \eqref{lemma:useful-ineq-ext-eq9}), 
$$
\tfrac{1}{t_k} = \tfrac{1}{T}(1 - \tfrac{T-t_k}{T})^{-1} \leq  \tfrac{1}{T}\,(1 + 2\tfrac{T-t_k}{T}),
$$
and $\pdf_0(t_k) = O(2^{-3\beta B})$ since $\pdf_0 \propto (\cl G')^3$.
\end{proof}

\section{Proof of Lemma \ref{lemma:bound-lpw-gaussian-vector}: ``Asymptotic Weighted $\ell_{p}$-Distortion''}
\label{proof-lemma-bound-lpw-gaussian-vector}

Before proving Lemma \ref{lemma:bound-lpw-gaussian-vector}, let us
show the following asymptotic equivalence.
\begin{lemma}
\label{lemma:bounding-a-tail-after-T}
Let $p\in \Nbb\setminus\{0\}$ and $\gamma>p-3$. 
\begin{equation}
  \label{eq:sum-pmom-assympt}
  \sum_{k=1}^{\cl B}\ [\cl G'(\omega_{k,p})]^{\gamma} \int_{\cl R_k}|t -
  \omega_{k,p}|^p\,\pdf_0(t)\,\ud t\ \simeq_B\  
\tfrac{2^{-pB}}{(p+1)\,2^p}\,\int_{\Rbb} [\cl G'(t)]^{\gamma-p}
\pdf_0(t)\,\ud t,
\end{equation}
\end{lemma}
\begin{proof} Let us use the threshold $T(B)$ defined in Lemma
  \ref{lemma:useful-ineq-ext} for splitting the sum
  \eqref{eq:sum-pmom-assympt} in two parts, \ie using the quantizer
  symmetry,
$$
\sum_{k=1}^{\cl B}\ [\cl G'(\omega_{k,p})]^{\gamma} \int_{\cl R_k}|t -
  \omega_{k,p}|^p\,\pdf_0(t)\,\ud t =\ 2\!\!\!\!\!\!\sum_{k:\ 0 \leq t_{k+1} <
T(B)} \ [\cl G'(\omega_{k,p})]^{\gamma} \int_{\cl R_k}|t -
\omega_{k,p}|^p\,\pdf_0(t)\,\ud t\quad +\quad \fs R,
$$
where the residual $\fs R$ reads 
\begin{align*}
\fs R&\eqdef\ 2\!\!\!\!\!\!\sum_{k:\ t_{k+1} \geq\, T(B)}\ [\cl G'(\omega_{k,p})]^{\gamma} \int_{\cl R_k}|t -
\omega_{k,p}|^p\,\pdf_0(t)\,\ud t,\\
&\,=\ 2\,[\cl G'(\omega_{k',p})]^{\gamma} \int_{\cl R_{k'}}|t -
\omega_{k',p}|^p\,\pdf_0(t)\,\ud t\ +\ 2\!\!\!\!\sum_{k:\ t_{k} \geq\, T(B)}\ [\cl G'(\omega_{k,p})]^{\gamma} \int_{\cl R_k}|t -
\omega_{k,p}|^p\,\pdf_0(t)\,\ud t,
\end{align*}
where $k'$ is such that $t_{k'} < T(B) \leq t_{k'+1}$.

From Lemma \ref{lemma:useful-ineq-ext}, we can easily bound this residual. We know from \eqref{lemma:useful-ineq-ext-eq1},
\eqref{lemma:useful-ineq-ext-eq2bis}, \eqref{lemma:useful-ineq-ext-eq10} and \eqref{lemma:useful-ineq-ext-eq11}
that, for all $k\in\ \{j: \omega_{j,p} \geq t_{j}\geq T(B)\}\cup\{k'\}$, 
$$
\ [\cl G'(\omega_{k,p})]^{\gamma} \int_{\cl R_k}|t -
\omega_{k,p}|^p\,\pdf_0(t)\,\ud t = O(2^{-\beta(\gamma+3)B} B^{-(p+1)/2}).
$$ 
However, \eqref{lemma:useful-ineq-ext-eq2} tells us that the sum in $\fs
R$ is made of no more than $1 + O\big(B^{-1/2}\,2^{(1-\beta)B}\big) = O\big(B^{-1/2}\,2^{(1-\beta)B}\big)$
terms, so that
$$
\fs R = O\big(B^{-(p+2)/2}\,2^{-(\beta (\gamma+4) - 1)B}\big).
$$ 

Let us now study the terms for which $0 \leq t_{k+1} \leq T(B)$. Using
\eqref{lemma:useful-ineq-ext-eq7} and \eqref{lemma:useful-ineq-ext-eq8} provides
\begin{align*}
&\sum_{k=1}^{\cl B}\ [\cl G'(\omega_{k,p})]^{\gamma} \int_{\cl R_k}|t -
  \omega_{k,p}|^p\,\pdf_0(t)\,\ud t\\
&\mathop{\simeq}_{B}\ 2\,
\sum_{k:\,0\leq t_{k+1}\leq T(B)}\ [\cl G'(\omega_{k,p})]^{\gamma} \tfrac{\gbinw_k^{p+1}}{(p+1)\,2^p} 
\,\pdf_0(\omega_{k,p})\ +\ \fs R\\
&\mathop{\simeq}_{B}\ 2 \tfrac{\binw^{p}}{(p+1)\,2^p}\,\sum_{k:\,0\leq t_{k+1}\leq T(B)}\ [\cl G'(\omega_{k,p})]^{\gamma-p}  
\,\pdf_0(\omega_{k,p})\,\gbinw_k\ +\ \fs R\\
&\mathop{\simeq}_{B}\ 
2\tfrac{2^{-pB}}{(p+1)\,2^p}\,\int_{0}^{T(B)} [\cl G'(t)]^{\gamma-p}
\pdf_0(t)\,\ud t\ +\ \fs R,
\end{align*}
where, knowing that $0 \leq t_{k+1}\leq T(B)$, we have also used \eqref{lemma:useful-ineq-ext-eq7} with
$p=0$ to see that $p_k=\int_{\cl R_k} \pdf_0(t)\,\ud t\simeq_B
\pdf_0(c')\gbinw_k$ for any $c'\in\cl R_k$.

Therefore, provided that $\beta(\gamma+4)\geq p+1$, which means that
$\gamma>p-3$ since $\beta<1$, the residual $\fs R$ decreases
faster than the first term in the right-hand side of last of the last equivalence relation, so that 
$$
\sum_{k=1}^{\cl B}\ [\cl G'(\omega_{k,p})]^{\gamma} \int_{\cl R_k}|t -
  \omega_{k,p}|^p\,\pdf_0(t)\,\ud t\ \mathop{\simeq}_{B}\ 
\tfrac{2^{-pB}}{(p+1)\,2^p}\,\int_{\Rbb} [\cl G'(t)]^{\gamma-p}
\pdf_0(t)\,\ud t,
$$
since $T(B)=\Theta(B^{1/2})$ by definition.
\end{proof}

With the three previous lemmata under our belts, we are now ready to
prove Lemma~\ref{lemma:bound-lpw-gaussian-vector}.

\begin{proof}[Proof of Lemma~\ref{lemma:bound-lpw-gaussian-vector}]
For $z_i\sim_{\iid} \cl N(0,\sigma_0^2)$ with pdf $\pdf_0$, using the SLLN
applied to $z_i$ conditionally on each quantization bin, we have
\begin{align*}
\|\cl Q_p[\bs z] - \bs z\|^p_{p,\bs w}&\eqdef\ \sum_{i=1}^M\ [\cl G'(\cl Q_p[z_i])]^{p-2}\,|z_i -
\cl Q_p[z_i]|^p,\\
&\mathop{\simeq}_{M}\ M\, \sum_{k=1}^{\cl B}\ [\cl G'(\omega_{k,p})]^{p-2} \int_{\cl R_k}|t -
\omega_{k,p}|^p\,\pdf_0(t)\,\ud t,
\end{align*}
where we used implicitly the quantizer symmetry in the last relation. 
This last relation is characterized by
Lemma~\ref{lemma:bounding-a-tail-after-T} by taking $n=p$ and $\gamma = p-2 >
p-3$, so that
\begin{align*}
\|\cl Q_p[\bs z] - \bs z\|^p_{p,\bs w}
&\mathop{\simeq}_{M,B}\ M\,\tfrac{2^{-pB}}{(p+1)\,2^p}\,\int_{\Rbb} [\cl G'(t)]^{-2}
\pdf_0(t)\,\ud t,\\
&\mathop{\simeq}_{M,B}\ M\, \tfrac{2^{-pB}}{(p+1)2^p}\,\fnorm{\pdf_0}_{1/3}.
\end{align*}
\end{proof}

\section{Proof of Lemma \ref{lem:strict-bounds-mu_p}: ``Gaussian $\ell_{p,w}$-Norm Expectation''}
\label{sec:proof-lemma-strict-bound}

First, the inequality $\E\norm{\bs \xi}_{p,\bs w} \leq (\E\norm{\bs
  \xi}^p_{p,\bs w})^{1/p}$ follows from the Jensen inequality applied on
the convex function $(\cdot)^p$ on $\Rbb_+$. Second, from our result in
\cite[Appendix C]{Jacques2010} it is easy to show that
$$
\E\norm{\bs \xi}_{p,\bs w} \geq (\E\norm{\bs
\xi}^p_{p,\bs w})^{1/p}\,\big(1\ +\ (\E\norm{\bs
\xi}^p_{p,\bs w})^{-2}\Var\norm{\bs
\xi}^p_{p,\bs w}\big)^{\inv{p}-1}.
$$
Moreover, $\E\norm{\bs \xi}^p_{p,\bs w} = \norm{\bs w}_p^p\,\E|\cl Z|^p$, while 
$$
\Var\norm{\bs
  \xi}^p_{p,\bs w}\ =\ \sum_i \Var |w_i \cl Z|^p = \|\bs w\|_{2p}^{2p}
  \Var|\cl Z|^p.
$$
Therefore, assuming CM weights,
\begin{align*}
\E\norm{\bs \xi}_{p,\bs w}/(\E\norm{\bs
\xi}^p_{p,\bs w})^{1/p}&\geq\ \big(1\ +\ (\rho^{\max}_{2p}/\rho^{\min}_p)^{2p}
M^{-1}(\E|\cl Z|^{p})^{-2}\Var|\cl Z|^{p}\big)^{\inv{p}-1}\\
&\geq\ \big(1\ +\ 2^{p+1}\,\theta_p^{p}\,M^{-1}\big)^{\inv{p}-1},
\end{align*}
since $\rho_{2p}^{\max}\leq \rho_{\infty}^{\max}$, and $(\E|\cl Z|^{p})^{-2}\Var|\cl Z|^{2p}< 2^{p+1}$ \cite{Jacques2010}.

\section{Proof of Proposition~\ref{prop:grip-gauss}: ``RIP$_{p,w}$ Matrix Existence''}
\label{sec:proof-prop-grip-gauss}

The proof proceeds simply by considering the Lipschitz function $F(\bs
u) = \|\bs u\|_{p,\bs w}$ and the expected value
$\mu=F(\bs\xi)$ for a random vector $\bs\xi\sim \cl N^M(0,1)$ in 
\cite[Appendix A]{Jacques2010}. The Lipschitz constant of $F$ is
$$
\lim_{\bs u{\mathop{\to}_{\neq}}\bs v}\big|F(\bs u) - F(\bs v)\big|\,/\,\|\bs u -\bs v\|\ =\ \|\bs w\|_{\infty}
\,\lambda_p,
$$ 
with $\lambda_p=\max(M^{(2-p)/2p},1)$ for $p\geq 1$. The value
$\mu=\E\norm{\bs \xi}_{p,\bs w}$ can be estimated thanks to Lemma
\ref{lem:strict-bounds-mu_p}. Indeed, it tells us that if $M\geq
2(2\theta_p)^p$, 
$$
\mu\ \geq\ \tinv{2}(\E\norm{\bs \xi}^p_{p,\bs
  w})^{1/p} \geq\ \tinv{2}\,\rho^{\min}_p\nu_p\,M^{1/p},
$$
with $\nu_p^p=\E|\cl Z|^p=2^{p/2}\pi^{-1/2}\Gamma(\tfrac{p+1}{2})$.

Inserting these results in \cite[Appendix A]{Jacques2010}, it is
easy to show that a matrix $\bs \Phi\sim \cl N^{M\times N}(0,1)$ is RIP$_{p,\bs w}(K,\delta,\mu)$ with a probability higher than $1-\eta$
if
$$
M^{2/\!\max(2,p)}\ \geq\ c\,\big(\tfrac{\rho^{\max}_\infty}{\delta\,\rho_p^{\min}}\big)^2
\big(K\log [e\tfrac{N}{K}(1+12\delta^{-1})]\ +\ \log\tfrac{2}{\eta}\big), 
$$
for some constant $c>0$.

\section{Proof of Proposition~\ref{prop:towards-quant-cons}: Dequantizing Reconstruction Error}
\label{proof:towards-quant-cons}

\begin{proof}
We have to bound $\epsilon_p/\E\|\bs \xi\|_{p,\bs w}$, with $\bs \xi\sim\cl
  N^{M}(0,1)$, when $M$ is large and under the HRA.  First, according
  to Lemma~\ref{lem:strict-bounds-mu_p}, using the SLLN and using the same decomposition than in the proof of
  Lemma~\ref{lemma:bound-lpw-gaussian-vector} with the threshold
  $T(B)$ (with $\beta = (p+1)/(p+2)$) and the bounds provided by
  Lemma~\ref{lemma:useful-ineq-ext}, we find almost surely
\begin{align*}
  \mu^p\ \eqdef\ (\E\|\bs \xi\|_{p,\bs w})^p\ &\mathop{\simeq}_{M}\
  \sum_{i=1}^M\,[\cl G'(\cl Q_p[z_i])]^{p-2} \E|\cl Z|^p\\
&\mathop{\simeq}_{M}\ M\,\E|\cl Z|^p\,\sum_{k:\,t_k\geq 0}
  p_k\,[\cl G'(\omega_{k,p})]^{p-2}.
\end{align*}
The sum in the last expression is characterized by
Lemma~\ref{lemma:bounding-a-tail-after-T} by setting inside
\eqref{eq:sum-pmom-assympt} $n=0$ and $\gamma = p-2$. This provides
\begin{align*}
\mu^p\ &\mathop{\simeq}_{M,B}\ M\, \E|\cl Z|^p\,\int_{\Rbb}
  [\cl G'(t)]^{p-2}\pdf_0(t)\,\ud t\\
&\mathop{\simeq}_{M,B}\ M\, \E|\cl Z|^p\,\big[\int_{\Rbb} \pdf_0^{1/3}(t)\big]^{2-p}\,\big[\int_{\Rbb}
\pdf_0^{(p+1)/3}(t)\,\ud t\big]. 
\end{align*}

Therefore, using the value $\epsilon_p$ defined in Lemma~\ref{lemma:bound-lpw-gaussian-vector},
$$
\tfrac{\epsilon^p}{\mu^p} \mathop{\simeq}_{B,M}
\tfrac{2^{-p(B+1)}}{(p+1)\,\E|\cl Z|^p}\,\fnorm{\pdf_0}^{(p+1)/3}_{1/3}\,\fnorm{\pdf_0}^{-(p+1)/3}_{(p+1)/3}
$$
However, for $\alpha > 0$,
$$
\fnorm{\pdf_0}^{\alpha}_{\alpha} \eqdef \int_{\Rbb} \pdf^\alpha_0(t)\,\ud t =
(2\pi\sigma_0^2)^{-\alpha/2}\,(2\pi\sigma_0^2/\alpha)^{1/2}\ \int_{\Rbb}
\gamma_{0,\sigma_0/\sqrt{\alpha}}(t)\,\ud t\ =\ (2\pi\sigma_0^2)^{(1-\alpha)/2}/\sqrt{\alpha}.
$$
Consequently, $\fnorm{\pdf_0}^{(p+1)/3}_{1/3} =
3^{(p+1)/2}\,(2\pi\sigma_0^2)^{(p+1)/3}$ and
$\fnorm{\pdf_0}^{(p+1)/3}_{(p+1)/3}=
(2\pi\sigma_0^2)^{(2-p)/6}/\sqrt{(p+1)/3}$, so that
$$
\tfrac{\epsilon^p}{\mu^p} \mathop{\simeq}_{B,M}
\tfrac{2^{-p(B+1)}}{\sqrt{p+1}\,\E|\cl Z|^p}\,(6\pi\sigma_0^2)^{p/2}
$$
Knowing that $(\E|\cl Z|^p)^{1/p}\geq c\,\sqrt{p+1}$ with
$c=8\sqrt{2}/(9\sqrt{e})$ \cite{Jacques2010}, we get
$$
\tfrac{\epsilon}{\mu}\ \lesssim_{B,M}\ c'\,2^{-B}\,
\,\tfrac{(p+1)^{-\inv{2p}}}{\sqrt{p+1}}\leq c'\,\frac{2^{-B}}{\sqrt{p+1}}.
$$
with $c'=(9/8)(e\pi/3)^{1/2}$.
\end{proof}

\section{Computation of the $\omega_{k,p}$}
\label{sec:comp-omeg-p}

This section describes a numerical procedure for efficiently computing
the $p$-optimal levels $\omega_{k,p}$ of a Gaussian source $\cl
N(0,1)$ for integer $p\geq 2$, defined by $\omega_{k,p} \eqdef
\argmin_{\lambda\in\cl R_k} \cl E_{k,p}(\lambda)$, where $ \cl
E_{k,p}(\lambda)= \int_{t_k}^{t_{k+1}} |t -\lambda|^p\
\gamma_{0,1}(t)\,\ud t. $ As $\cl E_{k,p}(\lambda)$ is strictly convex
and differentiable, the desired $\omega_{k,p}$ are the unique
stationary points satisfying $\cl E'_{k,p}(\omega_{k,p})=0$.

We compute the $\omega_{k,p}$ by Newton method, using
standard numerical quadrature for $\cl E'_{k,p}$ and $\cl
E''_{k,p}$. We handle the semi-infinite bins by replacing
$t_1=-\infty$ and $t_{\cl B}=\infty$ by -39 and +39, respectively
(chosen as the smallest integer $x$ so that $\gamma_{0,1}(x)=0$ when
evaluated in double precision floating point arithmetic). Given
quadrature weights $c_i$, we approximate $\cl E_{k,p}$ by $\tilde{\cl
  E}_{k,p}(\lambda) = \sum_{i=1}^N c_i \gamma_{0,1}(x_i)
|x_i-\lambda|^p$ with $x_i=t_k+(i-1)\Delta x$, where $\Delta x =
(t_{k+1}-t_k)/(N-1)$. We then have 
$\tilde{\cl E'}_{k,p}(\lambda) = \sum_{i=1}^N c_i \gamma_{0,1}(x_i)
p|x_i-\lambda|^{p-1}\sign(x_i-\lambda)$
and
$\tilde{\cl E''}_{k,p}(\lambda) = \sum_{i=1}^N c_i \gamma_{0,1}(x_i)
p(p-1)|x_i-\lambda|^{p-2}$.
We initialize with the midpoint for each of the finite bins, \ie set
$\lambda_k^{(0)}=(t_k+t_{k+1})/2$ for $2\leq k \leq \cl B-1$, and
$\lambda_1^{(0)}=t_2$, $\lambda_{\cl B}^{(0)}=t_{\cl B-1}$ for the
semi-infinite bins. For each $k$ we then iterate the Newton step $
\lambda_k^{(n)}=\lambda_k^{(n-1)} - \tilde{\cl
  E'}_{k,p}(\lambda_k^{(n-1)})/\tilde{\cl
  E}''_{k,p}(\lambda_k^{(n-1)}) $ until the convergence criterion
$|(\lambda_k^{n}-\lambda_k^{n-1})/\lambda_k^{n} |<10^{-15}$
is met. We used $c_i$ given by the fourth-order accurate Simpson's
rule, \eg $\bs c=(1,4,2,4 \hdots 2,4,1) \Delta x /3$, which yielded
empirically observed $O(N^{-4})$ convergence of the calculated
$w_{k,p}$. Results in this paper employed $N=10^4+1$ quadrature
points, sufficient to yield $w_{k,p}$ accurate to machine precision.

\end{document}